\documentclass[a4paper,UKenglish,cleveref,autoref]{lipics-v2019}

\usepackage{url}
\usepackage{graphicx}

\title{Reachability for updatable timed automata made faster and more effective}

\author{Paul Gastin}{LSV, ENS Paris-Saclay, CNRS, Universit\'e Paris--Saclay,
France}{paul.gastin@ens-paris-saclay.fr}{https://orcid.org/0000-0002-1313-7722}{}

\author{Sayan Mukherjee}{Chennai Mathematical Institute, India}{sayanm@cmi.ac.in}{https://orcid.org/0000-0001-6473-3172}{}

\author{B Srivathsan}{Chennai Mathematical Institute, India}{sri@cmi.ac.in}{https://orcid.org/0000-0003-2666-0691}{}

\authorrunning{P. Gastin, S. Mukherjee, and B. Srivathsan}

\Copyright{Paul Gastin, Sayan Mukherjee, B Srivathsan}



\category{}

\relatedversion{}

\supplement{}

\nolinenumbers

\hideLIPIcs

\usepackage{enumitem}
\usepackage{amssymb}
\usepackage{amsmath}
\usepackage{amsfonts}
\usepackage{multirow}
\usepackage{multicol}
\usepackage{wrapfig}
\usepackage{array}
\usepackage{tikz}
\usetikzlibrary{calc,decorations.pathmorphing}
\newcommand{\Aa}{\mathcal{A}}
\newcommand{\Bb}{\mathcal{B}}

\newcommand{\Gg}{\mathcal{G}}

\renewcommand{\d}{\delta}
\newcommand{\e}{\eta}

\newcommand{\Land}{\bigwedge}

\newcommand{\incl}{\subseteq}
\newcommand{\Rpos}{\mathbb{R}_{\ge 0}}
\newcommand{\Nat}{\mathbb{N}}
\newcommand{\Int}{\mathbb{Z}}

\newcommand{\PSPACE}{\textsc{Pspace}}
\newcommand{\lleq}{\mathrel{\triangleleft}}
\newcommand\sem[1]{{[\![ #1 ]\!]}}
\newcommand{\Valset}{\mathbb{V}}
\newcommand{\xra}{\xrightarrow}
\newcommand{\simmA}{\sqsubseteq}
\newcommand{\gsim}[1]{\sqsubseteq_{\scriptscriptstyle #1}}
\newcommand{\lug}{\gsim{G}}
\newcommand{\upinv}{\ensuremath{up}^{-1}}
\newcommand{\gq}{\gsim{\Gg(q)}}
\newcommand{\gqzero}{\gsim{\Gg(q_0)}}
\newcommand{\simg}{\gsim{g}}
\renewcommand{\wp}{\operatorname{pre}}
\newcommand{\gqp}{\gsim{\Gg(q')}}
\newcommand{\Mbar}{\overline{M}}
\newcommand{\Ggbar}{\overline{\Gg}}
\newcommand{\context}[1]{\overline{#1}}

\begin{document}

\maketitle

\begin{abstract}
  Updatable timed automata (UTA) are extensions of classic timed
  automata that allow special updates to clock variables, like $x:= x
  - 1$, $x := y + 2$, etc., on transitions.  Reachability for
  UTA is undecidable in general.  Various subclasses with decidable
  reachability have been studied.  A generic approach to UTA
  reachability consists of two phases: first, a static analysis of the
  automaton is performed to compute a set of clock constraints at each
  state; in the second phase, reachable sets of configurations, called
  zones, are enumerated.  In this work, we improve the algorithm for
  the static analysis.  Compared to the existing algorithm, our method
  computes smaller sets of constraints and guarantees termination for
  more UTA, making reachability faster and more effective. As the main
  application, we get an alternate proof of decidability and a more
  efficient algorithm for timed automata with bounded subtraction, a
  class of UTA widely used for modelling scheduling problems. We have
  implemented our procedure in the tool TChecker and conducted
  experiments that validate the benefits of our approach.
\end{abstract}

\section{Introduction}

Timed automata~\cite{Alur:TCS:1994} are finite automata equipped with
real-time variables called clocks. Values of the clock variables
increase at the same rate as time progresses. Transitions are guarded
by constraints over the clock variables. During a transition, the
value of a variable can be updated in several ways. In the classic
model, variables can be reset to $0$, written as a command $x := 0$ in
transitions. Generalizations of this involve $x := c$ with $c \ge 0$
or $x := y + d$ where $d$ is an arbitrary integer. Automata with these
more general updates are called \emph{Updatable Timed Automata
(UTA)}~\cite{Bouyer:2004:Updateable,Bouyer:2004:forwardanalysis}. The
updates provide a ``discrete jump'' facility during transitions. These
are useful syntactic constructs for modeling real-time systems and
have been used in several
studies~\cite{Fersman:InfComp:2007,Yi:2007:FORMATS:Schedulability-extended,Henzinger:1998:JCSS:Hybrid,Hatvani:2014:AVOCSjournal,Hourglass}.

On the one hand, variables with both a continuous and a discrete flow
offer modeling convenience. On the other hand, the discrete jumps are
powerful enough to simulate counter machines through the use of
$x := x + 1$ and $x := x - 1$ updates, in fact with zero time elapse
during the entire simulation~\cite{Bouyer:2004:Updateable}. This makes
reachability for this model undecidable. Various decidable subclasses
have been investigated over the
years~\cite{Bouyer:2004:Updateable,Fersman:InfComp:2007}. The most
common technique to prove decidability involves showing the existence
of a region automaton~\cite{Alur:TCS:1994}, which is a finite
automaton accepting the (untimed) sequences of actions that have a
timed run in the UTA.  Although this gives decidability, the algorithm
via the region construction is impractical due to the presence of
exponentially many regions. Practical algorithms in current tools like
UPPAAL~\cite{Larsen:1997:UPPAAL}, PAT~\cite{PAT-ModelChecker},
Theta~\cite{theta-fmcad2017} and TChecker~\cite{tchecker} work with
\emph{zones}, which are bigger sets of configurations than regions and
can be efficiently represented and manipulated using Difference-Bound
Matrices (DBMs)~\cite{Dill:1990:DBM}. Notably, these tools implement
zone based algorithms only for UTA with restricted updates $x : = c$
for $c \ge 0$, which behave similar to the reset $x := 0$. Most of the
efforts in making the zone based algorithm more efficient have
concentrated on this subclass of timed automata with only
resets~\cite{Behrmann:STTT:2006,Herbreteau:IandC:2016,Ocan:CAV:19:abstraction-refinement}.

Recently, we have presented a zone based algorithm for updatable timed
automata~\cite{CAV-19}. Due to the undecidability of the problem, it
cannot deal with the whole class of UTA. It however covers the
subclasses tabulated in~\cite{Bouyer:2004:Updateable}. The algorithm
consists of two phases: first, a static analysis of the automaton is
performed to compute a set of clock constraints at each state of the
automaton; in the second phase, reachable sets of configurations,
stored as zones, are enumerated.  None of these phases has a
guaranteed termination.  If the static analysis terminates, a
simulation relation between zones based on the constraints generated
in the static analysis can be used to guarantee termination of the
zone enumeration. Moreover, a smaller set of constraints in the static
analysis gives a coarser simulation which leads to a faster zone
enumeration.  The simulation used in \cite{CAV-19} lifts the efficient
$LU$-simulation~\cite{Behrmann:STTT:2006,Herbreteau:IandC:2016}
studied for diagonal-free reset-only timed automata to automata with
diagonal constraints and updates.

\subparagraph
{Contributions.} In this work, we strongly improve the
static analysis of \cite{CAV-19}. The new approach accumulates fewer
clock constraints and terminates for a wider class of UTA. In
particular, it terminates for \emph{timed automata with bounded
subtraction}, which was not the case before. This
class contains updates $x:= x - c$ with $c \ge 0$ along with
resets. However, an update $x:= x - c$ is allowed in a transition
only when there is a promise that each configuration that can take
this transition has a bounded $x$-value. This boundedness property
gives decidability thanks to a finite region equivalence. This
class has been used to model schedulability
problems~\cite{Fersman:InfComp:2007}, where updates $x := x - c$ have
been crucially used to model preemption. Thus, our new static analysis
allows to use efficient simulations during the zone enumeration for
this class.

At an algorithmic level, the new analysis is a slight modification of
the older one. However, this makes some of the technical questions
significantly harder: we show that deciding termination of the new
analysis can be done in polynomial-time if the constants in the guards
and updates are encoded in unary, whereas the problem is
$\PSPACE$-complete when the constants are encoded in binary. The older
analysis does not depend on the encoding of constants, and has a
polynomial-time algorithm for deciding termination.

For the experiments, the differences in the encoding and the hardness
result do not carry much importance. The static analysis is
implemented as a fixed-point iteration which can continue for a fixed
number of steps determined by the size of the automaton, or can be
stopped after a fixed time-out. We have implemented the new static
analysis in the open source tool TChecker~\cite{tchecker}.  We noticed
that the new method terminates and produces a result for more cases,
and when both methods produce a result, the new method is faster.

\subparagraph
{Related work.} Static analysis for timed automata without
diagonal constraints and with updates restricted to $x := c$ and
$x := y + c$ with $c \ge 0$ was studied in \cite{Behrmann:TACAS:2003}
in the context of $M$-simulations, which were implemented in earlier
versions of UPPAAL and KRONOS~\cite{KRONOS}. Latest tools implement a
more efficient
$LU$-simulation~\cite{Behrmann:STTT:2006,Herbreteau:IandC:2016}. Our
method clarifies how some optimizations of \cite{Behrmann:TACAS:2003}
can be lifted to the context of $LU$-simulations and more general
updates, and also provides additional
optimizations. TIMES~\cite{TIMES-Tool} is a tool for modeling
scheduling problems. It is mentioned in~\cite{Fersman:InfComp:2007}
that TIMES implements an algorithm using zones based on ``the UPPAAL
DBM library extended with a subtraction operator''. However, the exact
simulations used in the zone enumeration are not clear to us. A
different approach to reachability is presented in
\cite{Herbreteau:2013:CAV} where the constraints needed for simulation
are learnt during the zone enumeration directly. This potentially
gives more relevant constraints and hence coarser simulations. On the
flip side, it requires a sophisticated zone enumeration method with
observable overheads. Moreover \cite{Herbreteau:2013:CAV} deals with
timed automata without diagonal constraints and general
updates. Static analysis is lucrative since it is cheap, and maintains
the reachability procedure as two simple steps. Apart from
verification of UTA, studies on the expressive power of updates and
diagonal constraints have been carried out in
\cite{Bouyer:2004:Updateable,Bouyer:2005:Conciseness}. Timed register
automata~\cite{Bojanczyk:ICALP:2012} are a variant of UTA that have
been looked at in the context of canonical representations.

\subparagraph
{Organization.} Section~\ref{sec:preliminaries-new} gives the
preliminary definitions. Section~\ref{sec:better-gg-bounds} introduces
the new static analysis approach. Some classes of UTA where the new
static analysis can be applied are discussed in Section
\ref{sec:applications}. The subsequent Section~\ref{sec:termination}
discusses the termination problem for our proposed static
analysis. Section~\ref{sec:experiments} provides the results of our
experiments and Section~\ref{sec:benchmarks} contains details
about the models used for the experiments. We conclude with Section
\ref{sec:conclusion}.


\section{Preliminaries}
\label{sec:preliminaries-new}

We denote by $\mathbb{R}$ the set of reals, by $\Rpos$ the
non-negative reals, by $\Int$ the integers and by $\Nat$ the natural
numbers. Let $X$ be a finite set of variables over $\Rpos$ called
\emph{clocks}.  A \emph{valuation} is a function
$v\colon X \to \mathbb{R}_{\ge 0}$ that maps every clock to a
non-negative real number. For $\d \in \Rpos$ we define valuation
$v + \d$ as $(v + \d)(x) := v(x) + \d$. The set of valuations is
denoted by $\mathbb{V}$.

A \emph{non-diagonal constraint} is an expression of the form
$x \lleq c$ or $c \lleq x$, where $x \in X$, $c \in \Nat$ and
${\lleq}\in\{{<},{\leq}\}$, that is, $x \lleq 3$ stands for either
$x< 3$ or $x\le 3$.  A \emph{diagonal constraint} is an expression of
the form $x - y \lleq c$ or $c \lleq x - y$ where $x,y \in X$ are
clocks and $c \in \Nat$.  An \emph{atomic constraint} is either a
non-diagonal constraint or a diagonal constraint.  We also consider
two special atomic constraints $\top$ (true) and $\bot$ (false).  A
\emph{constraint} $\varphi$ is either an atomic constraint or a
conjunction of atomic constraints, generated by the following grammar:
$\varphi ::= \top \mid \bot \mid x \lleq c \mid c \lleq x \mid x - y
\lleq c \mid c \lleq x - y \mid \varphi \wedge \varphi$ with
$c \in \Nat$, ${\lleq}\in\{{<},{\leq}\}$.  Given a constraint
$\varphi$ and a valuation $v$, we write $v(\varphi)$ for the boolean
expression that we get by replacing every clock $x$ present in
$\varphi$ with the value $v(x)$. A valuation $v$ is said to
\emph{satisfy} a constraint $\varphi$, written as $v \models \varphi$,
if the expression $v(\varphi)$ evaluates to true.  For every valuation
$v$, we have $v \models \top$ and $v \not\models \bot$.  Given a
constraint $\varphi$, we define the set
$\sem{\varphi} := \{ v \mid v \models \varphi \}$.

An \emph{update} $up\colon \mathbb{V} \mapsto \mathbb{V}$ is a partial
function mapping valuations to valuations. The update $up$ is
specified by an \emph{atomic update} for each clock $x$, given as
either $x := c$ or $x := y + d$ where $c \in \Nat$, $d \in \Int$ and
$y \in X$ (is possibly equal to $x$). We write $up_x$ for the right
hand side of the atomic update of $x$, that is, $up_x$ is either $c$
or $y + d$. Note that we want $d$ to be an integer, since we allow for
decrementing clocks, and on the other hand $c \in \Nat$ since clock
values are always non-negative.  Given a valuation $v$ and an update
$up$, we define $v(up_x)$ to be $c$ or $v(y) + d$ depending on $up_x$
being $c$ or $y + d$. We say $up(v) \ge 0$ if $v(up_x)\ge 0$ for all
$x\in X$.  In this case the \emph{valuation} $up(v)\in\mathbb{V}$ is
defined by $up(v)(x)=v(up_x)$ for all $x\in X$. In general, due to the
presence of updates $up_x:= y +d$ with $d < 0$, the update may not
yield a clock valuation and for those valuations $v$, $up(v)$ is not
defined.  For example, if $v(x) = 5$ and $up_x = x - 10$ then $up(v)$
is undefined.  Hence, the domain of the partial function
$up\colon \mathbb{V}\to\mathbb{V}$ is the set of valuations $v$ such
that $up(v)\geq0$.  Updates can be used as transformations in timed
automata transitions.  An updatable timed automaton is an extension of
a classical timed automaton which allows updates of clocks on
transitions.

\begin{definition}
  An \emph{updatable timed automaton (UTA)} $\Aa = (Q, X, q_0, T, F)$
  is given by a finite set $Q$ of states, a finite set $X$ of clocks,
  an initial state $q_0$, a set $T$ of transitions and $F \incl Q$ of
  accepting states.  Transitions are of the form $(q, g, up, q')$
  where $g$ is a constraint (also called guard) and $up$ is an update,
  $q, q' \in Q$ are the source and target states respectively.
\end{definition}

Fix a UTA $\Aa := (Q, X, q_0, T, F)$ for the rest of this section. A
configuration of $\Aa$ is a pair $(q, v)$ with $q \in Q$ and
$v \in \Valset$. Semantics of $\Aa$ is given by a transition system
over its configurations. There are two kinds of transitions:
\emph{delay} and \emph{action}. For every configuration $(q, v)$ and
every $\d \in \Rpos$ there is a delay transition
$(q, v) \xra{\d} (q, v + \d)$. For every transition
$t:= (q, g, up, q')$ in the automaton, there is an action transition
$(q, v) \xra{t} (q', v')$ in the semantics if $v \models g$ ($v$
satisfies the guard), $up(v) \ge 0$ (the update on $v$ is defined) and
$v' = up(v)$. The initial configuration is $(q_0, v_0)$ with
$v_0(x) = 0$ for every clock $x$. We write
$(q, v) \xra{\d, t} (q',v')$ for the sequence of delay $\d$ and action
$t$ from $(q, v)$. A \emph{run} is an alternating sequence of delay
and action transitions starting from the initial configuration:
$(q_0, v_0) \xra{\d_1, t_1} (q_1, v_1) \xra{\d_2, t_2} \cdots
\xra{\d_{n}, t_{n}} (q_n, v_n)$. The run is accepting if $q_n \in F$.

The reachability problem for UTA asks if a given UTA has an accepting
run. This problem is undecidable in
general~\cite{Bouyer:2004:Updateable}. Various decidable fragments
with a $\PSPACE$-complete reachability procedure have been
studied~\cite{Bouyer:2004:Updateable,Maler:2006:Scheduling,CAV-19}. The
basic reachability procedure involves computing sets of reachable
configurations of the UTA stored as constraints which are popularly
called as \emph{zones}~\cite{Daws:TACAS:1998}. A zone is a set of
valuations given by a conjunction of atomic constraints $x \lleq c$,
$c \lleq x$, $x - y \lleq c$ and $c \lleq x - y$ with $c \in \Nat$ and
$x, y \in X$. For example $(x - y \le 5) \land (2 < x)$ is a
zone. Given a state-zone pair $(q,Z)$ (henceforth called a
\emph{node}) and a transition $t:= (q, g, up, q')$, the set of
valuations
$Z_t:= \{up(v) + \d \mid v \in Z, v \models g, up(v) \ge 0, \d \ge 0
\}$ is a zone. This is the set of valuations obtained from the $v$ in
$Z$ that satisfy the guard $g$ of the transition, get updated to
$up(v)$ and then undergo a delay $\d$.  The initial node $(q_0, Z_0)$
is obtained by delay from the initial configuration:
$Z_0 := \{ v_0 + \d \mid \d \ge 0 \}$ is a zone. This lays the
foundation for a reachability procedure: start with the initial node
$(q_0, Z_0)$; from each node $(q, Z)$ that is freshly seen, explore
the transitions $t:= (q, g, up, q')$ out of $q$ to compute resulting
nodes $(q', Z_t)$. If a pair $(q, Z)$ with $q \in F$ is visited then
the accepting state is reachable in the UTA. This na\"ive zone
enumeration might not terminate~\cite{Daws:TACAS:1998}. For
termination, \emph{simulations} between zones are used.

A simulation relation on the UTA semantics is a preorder relation (in
other words, a reflexive and transitive relation)
$(q, v) \simmA (q, v')$ between configurations having the same state
such that the relation is preserved (1) on delay:
$(q, v + \d) \simmA (q, v'+\d)$ for all $\d \in \Rpos$ and (2) on
actions: if $(q, v) \xra{t} (q_1, v_1)$, then
$(q, v') \xra{t} (q_1, v'_1)$ with $(q_1, v_1) \simmA (q_1, v'_1)$ for
all $t = (q, g, up, q_1)$. This relation gets naturally lifted to
zones: $(q, Z) \simmA (q, Z')$ if for all $v \in Z$ there exists a
$v' \in Z'$ such that $(q, v) \simmA (q, v')$. Intuitively, when
$(q, Z) \simmA (q, Z')$, all sequences of transitions enabled from
$(q, Z)$ are enabled from $(q, Z')$. Therefore, all control states
reachable from $(q, Z)$ are reachable from $(q, Z')$. This allows for
an optimization in the zone enumeration: a fresh node $(q, Z)$ is not
explored if there is an already visited node $(q, Z')$ with
$(q, Z) \simmA (q, Z')$.  A simulation $\simmA$ is said to be
\emph{finite} if in every sequence of the form
$(q, Z_0), (q, Z_1), \dots$ there are two nodes $(q, Z_i)$ and
$(q, Z_j)$ with $i < j$ such that $(q, Z_j) \simmA (q, Z_i)$. Using a
finite simulation in the reachability procedure ensures
termination. Various finite simulations have been studied in the
literature, the most prominent being $LU$-simulation
\cite{Behrmann:STTT:2006,Herbreteau:IandC:2016,Gastin:2018:CONCUR} and
more recently the $\Gg$-simulation~\cite{CAV-19}. In addition to
ensuring termination, one needs simulations which can quickly prune
the search.  One main focus of research in timed automata reachability
has been in finding finite simulations which are efficient in pruning
the search.

In a previous work~\cite{CAV-19}, we introduced a new simulation
relation for UTA, called the \emph{$\Gg$-simulation}. This relation is
parameterized by a set of constraints $\Gg(q)$ associated to every
state $q$ of the automaton. The sets $\Gg(q)$ are identified based on
the transition sequences from $q$. We now present the basic
definitions and properties of $\Gg$-simulation.  The presentation
differs from~\cite{CAV-19}, but the essence of the technical content
is the same.

\begin{definition}[$G$-preorder]
  \label{def:g-simulation}
  Given a finite or infinite set of constraints $G$, we say
  $v \lug v'$ if for every $\d \ge 0$, and every $\varphi \in G$:
  $v + \d \models \varphi$ implies $v' + \d \models \varphi$.
  
  We simply write $\gsim{\varphi}$ instead of $\gsim{\{\varphi\}}$
  when $G=\{\varphi\}$ is a singleton set.
\end{definition}

Directly from the definition of $\lug$, we get that the relation
$\lug$ is a preorder. The definition also entails the following useful
property: when $v \lug v'$, $v \models \varphi$ implies
$v' \models \varphi$ for all $\varphi \in G$.  This is a first step
towards getting a simulation on the UTA semantics. It says that all
guards that $v$ satisfies are satisfied by $v'$, and hence all
transitions enabled at $v$ will be enabled at $v'$ provided the
transition guards are present in $G$. Valuations get updated on
transitions and this property needs to be preserved over these
updates. This motivates the following definition. It gives a
constraint $\psi$ such that $v \gsim{\psi} v'$ will imply
$up(v) \gsim{\varphi} up(v')$.

\begin{definition}\label{def:upinv} Given an update $up$ and a
  constraint $\varphi$, we define $\upinv(\varphi)$ to be the
  constraint resulting by simultaneously substituting $up_x$ for $x$
  in $\varphi$: $\upinv(\varphi) := \varphi[up_x/x, \forall x \in X]$.
\end{definition}

For example, for $\varphi = x - y \lleq c$,
$up^{-1}(x - y \lleq c) = up_x - up_y \lleq c$.  Similarly,
$up^{-1}(x \lleq c) = up_x \lleq c$ and
$up^{-1}(c \lleq x) = c \lleq up_x$.  Note that, $up^{-1}(\varphi)$
need not be in the syntax defined by the grammar for constraints.
But, it can be easily rewritten to an equivalent constraint satisfying
this syntax.  For example: consider the constraint $x - y \lleq c$ and
the update $up_x = z + d$ and $up_y=y$, then
$up^{-1}(\varphi) 
= z + d - y \lleq c$, which is not syntactically a constraint.
However, it is equivalent to the constraint $z - y \lleq c - d$. If
$c - d < 0$, we further rewrite as $d - c \lleq y - z$.  It is also
useful to note that $up^{-1}(\varphi)$ may sometimes yield constraints
equivalent to $\top$ or $\bot$.  For example: if $\varphi = x \lleq c$
and $up_x = d$ with $d > c$, then the constraint $up^{-1}(\varphi)$ is
equivalent to $\bot$, similarly, if $d<c$ then $up^{-1}(\varphi)$ is
equivalent to $\top$.

\begin{lemma}
  \label{lem:upinv-imp-up}
  Given a constraint $\varphi$, an update $up$ and two valuations
  $v,v'$ such that $up(v)\geq0$ and $up(v')\geq0$, if
  $v \gsim{up^{-1}(\varphi)} v'$ then $up(v) \gsim{\varphi} up(v')$.
\end{lemma}
\begin{proof}
  Let $\d \ge 0$ and assume $up(v) + \d \models \varphi$. We need to
  show $up(v') + \d \models \varphi$.

  Suppose $\varphi$ is a diagonal constraint $x - y \lleq c$. Then,
  since $up(v) + \d \models \varphi$, we also have $up(v)$ (without
  the $ + \d$) satisfying $\varphi$. This entails
  $v \models \upinv(\varphi)$, by definition of $\upinv(\varphi)$. By
  hypothesis, we have $v \gsim{\upinv(\varphi)} v'$. Together with
  $v \models \upinv(\varphi)$, we get $v' \models
  \upinv(\varphi)$. This leads to $up(v') \models \varphi$. Again,
  since $\varphi$ is a diagonal, we get $up(v') + \d \models \varphi$.

  Suppose $\varphi$ is a non-diagonal constraint $x \lleq c$ or
  $c \lleq x$. We have respectively $up(v)(x) + \d \lleq c$ or
  $c \lleq up(v)(x) + \d$. If the update to $x$ in $up$ is $x := d$,
  then $up(v)(x) = d = up(v')(x)$ and hence $up(v') + \d$ also
  satisfies $\varphi$.  Else, the update is some $x:= y + d$. Now,
  depending on the constraint, we have either
  $v(y) + d + \d \lleq c $ or $c \lleq v(y) + d + \d$, equivalently
  $v(y) + \d \lleq c - d$ or $c - d \lleq v(y) + \d$. Notice that
  $\upinv(\varphi)$ is $y \lleq c - d$ or $c - d \lleq y$
  respectively. Use $v \gsim{\upinv(\varphi)} v'$ to infer that
  $v'(y) + \d \lleq c - d$ or $c - d \lleq v'(y) + \d$, leading to
  the required result $up(v') + \d \models \varphi$.
  
  When $\varphi = \bigwedge_i \varphi_i$, where each $\varphi_i$ is an
  atomic constraint, we have $up(v) + \d \models \varphi_i$ for every
  $i$. From the previous cases we get $up(v') + \d \models
  \varphi_i$. Therefore, $up(v') + \d \models \varphi$.
\end{proof}

\begin{definition}[$\Gg$-maps]\label{def:sound-g-bounds}
  Let $\Aa = (Q, X, q_0, T, F)$ be a UTA. A \emph{$\Gg$-map} $\Gg_\Aa$
  for UTA $\Aa$ is a tuple $(\Gg_\Aa(q))_{q \in Q}$ with each
  $\Gg_\Aa(q)$ being a set of atomic constraints, such that the
  following conditions are satisfied. For every transition
  $(q, g, up, q') \in T$:
  \begin{itemize}[nosep]
  \item every atomic constraint of $g$ belongs to $\Gg_\Aa(q)$,
  \item $\upinv(0\leq x) \in \Gg_\Aa(q)$ for every $x \in X$,
  \item $\upinv(\varphi) \in \Gg_\Aa(q)$ for every
    $\varphi \in \Gg_\Aa(q')$ \quad (henceforth called the
    \emph{propagation criterion})
  \end{itemize}
  When the UTA $\Aa$ is clear from the context, we write $\Gg$ instead
  of $\Gg_\Aa$.
\end{definition}

The propagation criterion allows to maintain the property described
after Definition~\ref{def:g-simulation} even after the update
occurring at transitions, and leads to a simulation relation on the
configurations of the corresponding UTA, thanks to
Lemma~\ref{lem:upinv-imp-up}.

\begin{definition}[$\Gg$-simulation]\label{def:g-simulation-with-g-map}
  Given a $\Gg$-map $\Gg$, the relation $\gsim{\Gg}$ on the UTA
  semantics defined as $(q, v) \gsim{\Gg} (q',v')$ whenever $q=q'$ and
  $ v \gq v'$, is called the $\Gg$-simulation.
\end{definition}

In general, an automaton may not have \emph{finite} $\Gg$-maps due to
the propagation criterion generating more and more constraints. When a
$\Gg$-map is finite, there is an algorithm to check
$(q, Z) \gq (q, Z')$. The fewer the constraints in a $\Gg(q)$, the
larger is the simulation $\gq$ (c.f.\
Definition~\ref{def:g-simulation}).  Hence there is more chance of
getting $(q, Z) \gq (q, Z')$ which in turn makes the enumeration more
efficient. Moreover, fewer constraints in $\Gg(q)$ give a quicker
simulation test $(q, Z) \gq (q, Z')$.  The goal therefore is to get a
$\Gg$-map as small as possible. Notice that if $\Gg_1$ and $\Gg_2$ are
$\Gg$-maps, then the map $\Gg_{min}$ defined as
$\Gg_{min}(q) := \Gg_1(q) \cap \Gg_2(q)$ is also a $\Gg$-map. A static
analysis of the automaton to get a $\Gg$-map is presented in
\cite{CAV-19}. The analysis performs an iterative fixed-point
computation which gives the smallest $\Gg$-map (for the pointwise
inclusion order) whenever it terminates. A procedure to detect if the
fixed-point iteration will terminate at all is also given in
\cite{CAV-19}.


\section{A new static analysis with reduced propagation of
  constraints}
\label{sec:better-gg-bounds}

In this section we give a refined propagation criterion, which cuts
short certain propagations. We start with a motivating
example. Figure~\ref{fig:non-terminating-G-computation} presents an
automaton and illustrates the fixed-point iteration computing the
smallest $\Gg$-map. Identity updates (like $y := y$) are not
explicitly shown. Only the newly added constraints at each step are
depicted. The first step adds constraints that meet the first two
conditions of Definition~\ref{def:sound-g-bounds}. Note that
$\upinv(0 \le y)$ is $0 \le y$ which is semantically equivalent to
$\top$. So we do not add it explicitly to the $\Gg$-maps.  Consider
two transitions $(q_0, v) \xra{t} (q_1, up(v))$ and
$(q_0, v') \xra{t} (q_1, up(v'))$ with
$t = (q_0, x \le 3, x:= x - 1 , q_1)$, and $up$ being $x := x -
1$. Suppose we require $up(v) \gsim{x - y < 1} up(v')$. By
Definition~\ref{def:g-simulation}, we need to satisfy the condition:
if $up(v) \models x - y < 1$, then $up(v') \models x - y <
1$. Rewriting in terms of $v$: if $v(x) - 1 - v(y) < 1$, then
$v'(x) - 1 - v'(y) < 1$. In other words, we need: if
$v \models x - y < 2$, then $v' \models x - y < 2$. This is achieved
by adding $x - y < 2$, the $\upinv(x - y < 1)$, to $\Gg(q_0)$ in the
second step. This is the essence of the propagation criterion of
Definition~\ref{def:sound-g-bounds}, which asks that for each
$\varphi \in \Gg(q_1)$, we have $\upinv(\varphi) \in \Gg(q_0)$. The
fixed-point computation iteratively ensures this criterion for each
edge of the automaton. As illustrated, the computation does not
terminate in Figure~\ref{fig:non-terminating-G-computation}. There are
three sources of increasing constants: (1) $x \le 3$,
$x \le 4, \dots$, (2) $1 \le x$, $2 \le x, \dots$ and (3)
$x - y < 1, x - y < 2, \dots$.

\begin{figure}[t]
  \centering
  \begin{tikzpicture}[state/.style={draw, circle, inner sep=2pt,
      minimum size = 4mm}]
    \begin{scope}[every node/.style={state}]
      \node (0) at (0,0) {\scriptsize $q_0$}; \node (1) at (3,0)
      {\scriptsize $q_1$}; \node (2) at (6,0) {\scriptsize $q_2$};
    \end{scope}
    \begin{scope}[->, >=stealth]
      \draw (-0.7,0) to (0); \draw (0) to [bend left=30] (1); \draw
      (1) to [bend left=30] (0); \draw (1) to (2);
    \end{scope}
    \node at (1.5,0.7) {\scriptsize $x \le 3$}; \node at (1.5,0.3)
    {\scriptsize $x:=x-1$}; \node at (4.5, 0.2) {\scriptsize
      $x - y < 1$};
  \end{tikzpicture}

  \bigskip

  \scriptsize
  \begin{tikzpicture}[state/.style={rectangle, inner sep=2pt, minimum
      size = 4mm, align=left}]
    \begin{scope}[every node/.style={state}]
      \node (-1) at (-1,0) {$\Gg(q_0) = $ \\ $\Gg(q_1) = $ \\
        $\Gg(q_2) = $};
      \node (1) at (0.5,0) {$ \{ x \le 3, ~1 \le x \} $ \\
        $ \{ x - y < 1 \} $ \\ $ \{ \} $}; \node (2) at (3.2,0)
      {$ \{ ~\dots, ~x - y < 2 \} $ \\
        $ \{ ~\dots, ~x \le 3, ~1 \le x\} $ \\ $ \{ \} $}; \node (3)
      at
      (6.3,0) {$ \{~\dots, ~x \le 4,  ~2 \le x \} $ \\
        $ \{~\dots,~x - y < 2 \} $ \\ $ \{ \} $};

      \node (4) at (9.6,0)
      {$ \{~\dots, ~x - y < 3 \} $ \\
        $ \{~\dots, ~x \le 4, ~2 \le x \} $ \\
        $ \{ \} $}; \node (7) at (11.5, 0) {$\dots$};
    \end{scope}
    \begin{scope}[->]
      \draw (1) to (2); \draw (2) to (3); \draw (3) to (4); \draw (4)
      to (7);
    \end{scope}
  \end{tikzpicture}
  \caption{Example automaton for which the $\Gg$-map computation of
    \cite{CAV-19} does not terminate}
  \label{fig:non-terminating-G-computation}
\end{figure}

We claim that this conservative propagation is unnecessary to get the
required simulation.  Suppose $v \gqzero v'$ and
$(q_0, v) \xra{t} (q_1, up(v))$, with
$t := (q_0, x \le 3, x:= x - 1, q_1)$.  Since $t$ is enabled at $v$,
we have $v(x) \le 3$, hence $v'(x)\le 3$ since guard $x \le 3$ is
present in $\Gg(q_0)$.  We get $v,v' \models x - y \le 3$ as $y \ge 0$
for all valuations. The presence of $x - y < 4, x - y < 5, \dots$ at
$\Gg(q_0)$ is useless as both $v, v'$ already satisfy these
guards. Stopping the propagation of $x - y < 3$ from $\Gg(q_1)$ will
cut the infinite propagation due to (3). A similar reasoning cuts the
propagation of $x \le 2$ from $\Gg(q_1)$ and stops (1). The remaining
source (2) is trickier, but it can still be eliminated. Here is the
main idea.  Consider a constraint $3 \le x \in \Gg(q_0)$ which
propagates unchanged to $\Gg(q_1)$ and then back to $\Gg(q_0)$ as
$up^{-1}(3\leq x)=4 \le x$. This propagation can be cut since
$v\gsim{3\leq x}v'$ already ensures $v\gsim{4 \le x}v'$ for the
valuations that are \emph{relevant}: the ones that satisfy the guard
$x \le 3$ of $t$.  Indeed, $v,v'\models x\le 3$ and
$v\gsim{3\leq x}v'$ implies $v(x)\le v'(x)$ which in turn implies
$v\gsim{4\leq x}v'$.  Overall, it can be shown that
$\Gg(q_0) = \{ x \le 3, ~3 \le x, ~x - y < 2, ~x -y < 3 \}$ and
$\Gg(q_1) = \{x - y < 1 \} \cup \Gg(q_0)$ suffices for the
$\Gg$-simulation.

\subparagraph {Taking guards into account for propagations.}  The
propagation criterion of Definition~\ref{def:sound-g-bounds} which is
responsible for non-termination, is oblivious to the guard that is
present in the transition. We will now present a new propagation
criterion that takes the guard into account and cuts out certain
irrelevant constraints.  Consider a transition $(q, g, up, q')$ and a
constraint $\varphi \in \Gg(q')$. All we require is a constraint
$\psi \in \Gg(q)$ such that $v \gsim{\psi} v'$ \emph{and}
$v \models g$ implies $up(v) \gsim{\varphi} up(v')$. The additional
``and $v \models g$'' was missing in the intuition behind the previous
propagation. Of course, setting $\psi := \upinv(\varphi)$ is
sufficient.  However, the goal is to either eliminate the need for
$\psi$ or find a $\psi$ with a smaller constant compared to
$\upinv(\varphi)$. We will see that in many cases, we can even get the
former, when we plug in the ``and $v \models g$''.

\begin{definition}[pre of an atomic constraint $\varphi$ under a
  ``guard-update'' pair $(g,up)$]
  \label{def:weakest-precondition}
  Let $(g,up)$ be a pair of a guard and an update.  For a constraint
  $\varphi$ we define $\wp(\varphi, g, up)$ to be an atomic constraint
  as given by Table~\ref{tab:optimizations-main}, when $g$ and
  $\upinv(\varphi)$ satisfy corresponding conditions. When the
  conditions of Table~\ref{tab:optimizations-main} do not apply,
  $\wp(\varphi, g, up) = up^{-1}(\varphi)$.

  For a set of constraints $\Gg$, we define $\wp(\Gg, g, up)$ to be
  the set $\bigcup_{\varphi \in \Gg} \{\wp(\varphi, g, up)\}$.
\end{definition}

\begin{table}[t]
  \centering
  \begin{tabular}{|c|c|c|c|}
    \hline
    & $\upinv(\varphi)$ & $g$ contains  & $\wp(\varphi, g, up)$ \\
    \hline
    1. & $x \lleq d$ & $x \lleq_1 c$ & $\top$ \\
    \hline
    2. & $d \lleq x$ & $x \lleq_1 c$ with $c<d$ &
                                                  $c\leq x$ \\
    \hline
    \multirow{2}{*}{3.} & \multirow{2}{*}{$x - y \lleq d~$ or $~d
                          \lleq x - y$ } & $x
                                           \lleq_1
                                           c~$ or $~x
                                           - y
                                           \lleq_1
                                           c~$ or $~e
                                           \lleq_1 x
                                           - y$
                                        &
                                          \multirow{2}{*}{$\top$}
    \\
    & &  s.t. $c < d < e$  & \\
    \hline
  \end{tabular}
  \smallskip
  \caption{Cases where $\upinv(\varphi)$ can be eliminated or replaced
    by a constraint with a smaller constant. We write $\lleq$ and
    $\lleq_1$ to insist that the operator $\lleq$ need not be same as
    the operator $\lleq_1$.}
  \label{tab:optimizations-main}
\end{table}

Our aim is to replace the $\upinv(\varphi)$ in the older propagation
criterion with $\wp(\varphi, g, up)$.  Before showing the correctness
of this approach, we state a useful lemma that follows directly from
the definition of $\Gg$-simulation.

\begin{lemma}\label{lem:g-sim-non-diagonal-LU-like} 
  Let $v, v'$ be valuations.
  \begin{itemize}[nosep]
  \item $v \gsim{x \lleq d} v'$ iff either $v \not\models x \lleq d$
    or $v'(x) \le v(x)$
  \item $v \gsim{d \lleq x} v'$ iff either $v' \models d \lleq x$ or
    $v(x) \le v'(x)$
  \end{itemize}
\end{lemma}
\begin{proof}
  \begin{itemize}
  \item Let $v \gsim{x \lleq d} v'$. By
    Definition~\ref{def:g-simulation}: if $v'(x) > v(x)$, then
    $v \not\models x \lleq d$. Since otherwise, there is a $\d \ge 0$
    such that $v + \d \models x \lleq d$, but
    $v' + \d \not\models x \lleq d$. For the converse, if
    $v \not \models x \lleq d$, then $v \gsim{x \lleq d} v'$
    vacuously. If $v'(x) \le v(x)$ then for every $\d \ge 0$, we have
    $v + \d \models x \lleq d$ implies $v' + \d \models x \lleq d$.

  \item Let $v \gsim{d \lleq x} v'$. Once again, by
    Definition~\ref{def:g-simulation}: if $v'(x) < v(x)$, then
    $v' \models d \lleq x$. This gives the forward implication. For
    the converse, if $v' \models d \lleq x$, then
    $v \gsim{d \lleq x} v'$ is vacuously true, and when
    $v(x) \le v'(x)$, the condition $v + \d \models d \lleq x$ implies
    $v' + \d \models d \lleq x$ is satisfied.

  \end{itemize}
\end{proof}

Readers familiar with the $LU$-simulation for diagonal-free
automata~\cite{Herbreteau:IandC:2016} may recognize that the above
lemma is almost an alternate formulation of the $LU$-simulation.  The
lemma makes a finer distinction between $<$ and $\le$ in the
constraints whereas $LU$ does not.

The next proposition allows to replace the $\upinv(\varphi)$ in
Definition~\ref{def:sound-g-bounds} by $\wp(\varphi, g, up)$ to get
smaller sets of constraints at each $q$ that still preserve the
simulation. We write $v \simg v'$ for $v \gsim{C_g} v'$, where $C_g$
is the set of atomic constraints in $g$.

\begin{proposition}
  \label{prop:pre-is-sound}
  Let $(g,up)$ be a guard-update pair, $v,v'$ be valuations such that
  $v \models g$ and $v \simg v'$, and $\varphi$ be an atomic
  constraint. Then,
  $v \sqsubseteq_{\scriptscriptstyle{\wp(\varphi,g,up)}} v'$ implies
  $v \gsim{\upinv(\varphi)} v'$.
\end{proposition}
\begin{proof}
  When $\wp(\varphi,g,up) = up^{-1}(\varphi)$, there is nothing to
  prove.  We will now prove the theorem for the combinations given in
  Table~\ref{tab:optimizations-main}.

  \emph{(Case 1).} From the hypothesis $v \gsim{g} v'$, we get
  $v \gsim{x \lleq_1 c} v'$. From the other hypothesis $v \models g$,
  we get $v \models x \lleq_1 c$. Therefore, by using the formulation
  of $v \gsim{x \lleq_1 c} v'$ from
  Lemma~\ref{lem:g-sim-non-diagonal-LU-like}, we get $v'(x) \le
  v(x)$. This entails $v \gsim{x \lleq d} v'$ for all upper bounded
  guards, once again from Lemma~\ref{lem:g-sim-non-diagonal-LU-like}.

  \emph{(Case 2).}  We have $\wp(\varphi,g,up) = c \leq x$ and
  $c < d$.  Moreover, as guard $g$ contains $x \lleq_1 c$, we have
  $v'(x) \le v(x)$ as in Case 1.  Since $v$ satisfies the guard, we
  get: $v'(x) \le v(x) \le c < d$.  From
  Lemma~\ref{lem:g-sim-non-diagonal-LU-like}, for such valuations,
  $v \gsim{c \leq x} v'$ implies $v'(x) = v(x)$.  Hence
  $v \gsim{d \lleq x} v'$.

  \emph{(Case 3).} There are sub-cases depending on whether the guard
  contains a non-diagonal constraint or the diagonal constraints. When
  the guard contains $x \lleq_1 c$, we have $v'(x) \le v(x) \le c$ as
  above. Hence $v'(x - y) \le c$ and $v(x - y) \le c$. Since we are
  given that $c < d$, both $v$ and $v'$ satisfy the diagonal
  constraint $x - y \lleq d$ and neither of them satisfies
  $d \lleq x - y$.  Notice that time elapse preserves the satisfaction
  of diagonal constraints as for every valuation $u$,
  $(u + d)(x - y) = u(x - y)$. From Definition~\ref{def:g-simulation},
  the $\Gg$-simulation for a diagonal constraint $\psi$ is satisfied
  if $v \not \models \psi$ or $v' \models \psi$. Hence,
  $v \gsim{x - y \lleq d} v'$ and $v \gsim{d \lleq x - y} v'$.

  For the other sub-cases of the guard containing $x - y \lleq_1 c$ or
  $e \lleq_1 x - y$, the hypotheses $v \models g$, $v\simg v'$ and the
  fact that $c < d < e$ ensure the same effect, that either $v$ does
  not satisfy the diagonal constraint $\upinv(\varphi)$ or $v'$
  does. Hence, by definition $v \gsim{\upinv(\varphi)} v'$.
\end{proof}

\begin{definition}[Reduced $\Gg$-maps]\label{def:reduced-g-bounds}
  A $\Gg$-map is said to be \emph{reduced} if for every transition
  $(q, g, up, q')$:
  \begin{itemize}[nosep]
  \item every atomic constraint of $g$ belongs to $\Gg(q)$,
  \item $\wp(0\leq x,g,up) \in \Gg(q)$ for every $x \in X$, and
  \item $\wp(\varphi, g, up) \in \Gg(q)$ for every
    $\varphi \in \Gg(q')$ \quad (reduced propagation)
  \end{itemize}
\end{definition}

Recall the definition of $\Gg$-simulation of
Definition~\ref{def:g-simulation-with-g-map}.  This is a relation
$\lug$ defined as $(q, v) \lug (q', v')$ whenever $q=q'$ and
$v \gsim{\Gg(q)} v'$.  The next theorem says that this relation stays
a simulation even when the $\Gg$-map is reduced.

\begin{theorem}
  \label{thm:pre-gives-simulation}
  Let $(\Gg(q))_ {q \in Q }$ be a reduced $\Gg$-map.  The relation
  $\gsim{\Gg}$ is a simulation.
\end{theorem}
\begin{proof}
  Let $(q, v) \gq (q, v')$. By Definition~\ref{def:g-simulation}, we
  have $(q, v + \d) \gq (q, v'+ \d)$. Now, suppose
  $(q, v) \xra{t} (q_1, v_1)$ with $t = (q, g, up, q_1)$. Since $v$
  satisfies guard $g$ and every atomic constraint of $g$ is present in
  $\Gg(q)$, we have $v'$ satisfying $g$. Moreover, $up(v) \ge
  0$. Since $\wp(0 \le x, g, up)$ is present in $\Gg(q)$ for all $x$,
  Proposition~\ref{prop:pre-is-sound} and Lemma~\ref{lem:upinv-imp-up}
  give $up(v) \sqsubseteq_{\{0 \le x, x \in X\} } up(v')$. This shows
  that $up(v') \ge 0$. For every $\varphi \in \Gg(q')$, we have
  $\wp(\varphi, g, up) \in \Gg(q)$. Once again,
  Proposition~\ref{prop:pre-is-sound} and Lemma~\ref{lem:upinv-imp-up}
  give $up(v) \gqp up(v')$. This proves that $\gsim{\Gg}$ is a
  simulation.
\end{proof}

As in the case of (non-reduced) $\Gg$-maps, notice that if $\Gg_1$ and
$\Gg_2$ are reduced $\Gg$-maps, the map $\Gg_{min}$ given by
$\Gg_{min}(q) = \Gg_1(q) \cap \Gg_2(q)$ is a reduced $\Gg$-map. There
is therefore a smallest reduced $\Gg$-map, given by the pointwise
intersection of all reduced $\Gg$-maps.

\begin{lemma}
  \label{lem:reduced-smallest-is-least-fixed-point}
  The smallest reduced $\Gg$-map with respect to pointwise inclusion
  is the least fixed-point of the following system of equations:
  $$
  \Gg(q)=\hspace{-5mm}\bigcup_{(q, g, up, q')}\hspace{-4mm}
  \{\text{atomic constraints of $g$}\} \cup \{\wp(0\leq x,g,up) \mid x
  \in X \} \cup \{\wp(\varphi, g, up) \mid \varphi \in \Gg(q')\}
  $$
\end{lemma}
\begin{proof}
  Every solution to this system of equations is a reduced $\Gg$-map
  and every reduced $\Gg$-map is a solution to this system.
\end{proof}

The smallest reduced $\Gg$-map can be computed by a standard Kleene
iteration. For every state $q$ and every $i \ge 0$:
\begin{align*}
  \Gg^0(q)  & ~= ~ \bigcup_{(q, g, up, q')} \{\text{atomic constraints of
              $g$}\} \cup \{\wp(0\leq x,g,up) \mid x \in X\} \\
  \Gg^{i+1}(q) & ~= ~\Gg^i(q) ~\cup~ \bigcup_{(q, g, up, q')} \{
                 \wp(\varphi, g, up) \mid \varphi \in \Gg^i(q')\}
\end{align*}
When $\Gg^{k+1} = \Gg^k$, a fixed-point has been found and $\Gg^k$ is
a reduced map satisfying
Definition~\ref{def:reduced-g-bounds}. Moreover, $\Gg^k$ gives the
least fixed-point to the system of equations of
Lemma~\ref{lem:reduced-smallest-is-least-fixed-point} and hence
$\Gg^k$ is the smallest reduced $\Gg$-map. When $\Gg^{i+1} \neq \Gg^i$
for all $i$, the least fixed-point is infinite and no reduced
$\Gg$-map for the automaton can be finite. For instance, if in the UTA
of Figure~\ref{fig:non-terminating-G-computation}, the guard $x \le 3$
is removed, the smallest reduced $\Gg$-map will be infinite, and the
fixed-point will continue forever, each iteration producing an
$x - y < c$ with increasing constants $c$.

It is not clear apriori how to detect whether the fixed-point
computation will terminate, or will continue forever. For the
non-reduced $\Gg$-maps, \cite{CAV-19} gives an algorithm that runs the
fixed-point computation (using $\upinv$ instead of $\wp$) for a
bounded number of steps and determines whether the computation will be
non-terminating by looking for a certain witness. The reduced
$\Gg$-map fixed-point is different due to
Table~\ref{tab:optimizations-main}, as certain propagations are
disallowed (Cases 1 and 3), or truncated to a constant determined by
the guard (Case 2). These optimizations are responsible for giving
finite $\Gg$-maps even when the non-reduced $\Gg$-maps are
infinite. This makes the termination analysis significantly more
involved.  We postpone this discussion to
Section~\ref{sec:termination}.  In the next section, we identify some
sufficient conditions that make the reduced $\Gg$-maps finite and
describe how it leads to new applications.  These observations throw
more light on the mechanics of the reduced $\Gg$-computation and
provide a preparation to the more technical
Section~\ref{sec:termination}.


\section{Applications of the reduced
  propagation}\label{sec:applications}

We exhibit three subclasses of UTA for which the reduced $\Gg$-maps
are superior than the non-reduced $\Gg$-maps: either reduced
$\Gg$-maps are finite whereas non-reduced $\Gg$-maps are not
guaranteed to be finite, or when both are 
finite, the reduced $\Gg$-map gives a bigger simulation.

\subparagraph
{Timed automata with bounded subtraction.}
Timed automata with diagonal constraints and updates restricted to
classic resets $x := 0$ and subtractions $x:= x - c$ with $c \geq 0$
have been used for modeling certain scheduling
problems~\cite{Fersman:InfComp:2007}. Reachability is undecidable for
this restricted update model~\cite{Bouyer:2004:Updateable}.  An
important result in~\cite{Fersman:InfComp:2007} is that reachability
is decidable for a subclass called timed automata with \emph{bounded
  subtraction}, and this decidability is used for answering the
schedulability questions. Proof of decidability proceeds by
constructing a region equivalence based on a maximum constant derived
from the automaton. We prove that timed automata with bounded
subtraction have finite reduced $\Gg$-maps. This gives an alternate
proof of decidability and a zone-based algorithm using
$\Gg$-simulation for this class of automata. This exercise also brings
out the significance of reduced $\Gg$-maps: without the reduced
computation, we cannot conclude finiteness.

\begin{definition}[Timed Automata with Bounded
  Subtraction~\cite{Fersman:InfComp:2007}]
  A timed automaton with ``subtraction'' is an updatable timed
  automaton with updates restricted to the form $x := 0$ and
  $x:= x - c$ for $c \geq 0$.  Guards contain both diagonal and
  non-diagonal constraints.

  A timed automaton with ``bounded subtraction'' is a timed automaton
  with subtraction such that there is a constant $M_x$ for each clock
  $x$ satisfying the following property for all its reachable
  configurations $(q,v)$: if there exists a transition $(q,g,up,q')$
  such that $v \models g$ and $up_x = x - c$ with $c>0$, then
  $v(x) \le M_x$.
\end{definition}

It is shown in~\cite{Fersman:InfComp:2007} that reachability is
decidable for timed automata with bounded subtraction when the bounds
$M_x$ are known. This definition of bounded subtraction puts a
semantic restriction over timed automata. Indeed, reachability is
decidable only when the bounds $M_x$ are apriori known. The following
is a syntactically restricted class of timed automata, that captures
the bounded subtraction model when the bounds $M_x$ are given.

\begin{definition}[Timed Automata with Syntactically Bounded
  Subtraction] \label{def:syntactically-bdd-subtraction} This is a
  timed automaton with subtraction such that, for every transition
  $(q,g,up,q')$ and clock $x$, if $up_x = x - c$ with $c>0$ then the
  guard $g$ contains an upper bound constraint $x \lleq c'$ for some
  $c' \in \Nat$.
\end{definition}

\begin{lemma}
  \label{lem:bdd-to-syntactically-bdd}
  For every timed automaton with bounded subtraction $\Aa'$ where the
  bound $M_x$ for every clock $x$ is known, there exists a timed
  automaton with syntactically bounded subtraction $\Aa$ such that the
  runs of $\Aa$ and $\Aa'$ are the same.
\end{lemma}
\begin{proof}
  Given $\Aa'$ we construct $\Aa$ as follows: for every transition
  $t_{\Aa'} = (q,g,up,q')$ change the guard of this transition to
  $g \wedge (\Land_{x \in B} x \le M_x)$ where $B$ is the set of
  clocks $x$ with update of the form $x:= x - c$ with $c>0$.
\end{proof}

\begin{theorem}
  \label{thm:reduced-g-bounds-finite-syntactically-bdd}
  The smallest reduced $\Gg$-maps are finite for timed automata with
  syntactically bounded subtraction.
\end{theorem}
\begin{proof}
  Let $\Mbar$ 
  be the maximum constant appearing among the guards and updates of
  the given automaton. Define $\Ggbar$ to be the (finite) set of all
  atomic constraints with constant at most $\Mbar$. We will show that
  the \emph{finite} map $\Gg$ assigning $\Gg(q) = \Ggbar$ for all $q$
  is a reduced $\Gg$-map. This then proves the theorem.

  The first two conditions of Definition~\ref{def:reduced-g-bounds}
  are trivially true. It remains to show that
  $\wp(\Ggbar,g,up) \subseteq \Ggbar$ for every transition
  $(q, g, up, q')$. Choose a constraint $\varphi \in \Ggbar$. Note
  that $\wp(\varphi, g, up)$ is a constraint having a larger constant
  than $\varphi$ only if $up$ contains subtractions (since the other
  possible update is only a reset to $0$ in this class). Thus, if $up$
  does not contain subtractions, from the construction of $\Ggbar$ it
  follows that $\wp(\varphi,g,up) \subseteq \Ggbar$. Now, if
  $up_x = x - c$ for some clock $x$ and $c>0$, then $g$ contains
  $x \lleq_1 c_1$ by definition.  If $\upinv(\varphi)$ is some
  $x \lleq d$, then Case 1 of Table~\ref{tab:optimizations-main} gives
  $\wp(\varphi, g, up) = \top$.  If $\upinv(\varphi)$ is $d \lleq x$,
  from Case 2 of the table, we have $\wp(\varphi,g,up)=c_1\leq x$ or
  $\wp(\varphi,g,up)=d\lleq x$ with $d\leq c_1$, which are both
  present in $\Ggbar$ by construction.

  Finally, assume that $\upinv(\varphi)$ is a diagonal constraint
  $x - y \lleq d$ or $d \lleq x - y$ and Case 3 of
  Table~\ref{tab:optimizations-main} does not apply.  We have
  $up_x=x-c_1$ with $c_1\geq0$ and $up_y=y-c_2$ with $c_2\geq0$ (a
  reset for $x$ or $y$ is not possible).  Moreover, if $c_1>0$ (resp.\
  $c_2>0$) then $g$ contains some $x\lleq_1 c'_1$ (resp.\
  $y\lleq_2 c'_2$).  If $c_1>0$ then, since Case 3 does not apply, we
  get $d\leq c'_1\leq \Mbar$ and $up^{-1}(\varphi)$ belongs to
  $\Ggbar$.  If $c_1=0$ and $c_2>0$ then the constraint $\varphi$ is
  respectively $x-y \lleq d+c_2$ or $d+c_2 \lleq x-y$.  Since
  $0\leq d<d+c_2\leq \Mbar$, the constraint $\upinv(\varphi)$ is
  already in $\Ggbar$.
\end{proof}

Lemma~\ref{lem:bdd-to-syntactically-bdd} and
Theorem~\ref{thm:reduced-g-bounds-finite-syntactically-bdd} give an
alternate proof of decidability and more importantly a zone based
algorithm with optimized simulations for this model. The definition of
timed automata with bounded subtraction can be seamlessly extended to
include updates $x:= y - c$ where $c \ge 0$ and $x, y$ are potentially
different clocks. Definition~\ref{def:syntactically-bdd-subtraction},
Lemma~\ref{lem:bdd-to-syntactically-bdd} and Theorem
\ref{thm:reduced-g-bounds-finite-syntactically-bdd} can suitably be
modified to use $y \lleq c'$ instead of $x \lleq c'$. This preserves
the decidability, with similar proofs, even for this extended class.

\subparagraph
{Clock bounded reachability.} 
Inspired by
Theorem~\ref{thm:reduced-g-bounds-finite-syntactically-bdd}, we
consider the problem of clock-bounded reachability: given UTA and a
bound $B \ge 0$, does there exist an accepting run
$(q_0, v_0) \xra{} (q_1, v_1) \xra{} \cdots (q_n, v_n)$ where
$v_i(x) \le B$ for all $i$ and all clocks $x$? This problem is
decidable for the entire class of UTA. The algorithm starts with a
modified zone enumeration: each new zone is intersected with
$\Land_{x} x \le B$ before further exploration. This way, only the
reachable configurations within the given bound are stored. The number
of bounded zones is finite. Hence the enumeration will terminate
without the use of any simulations. On the other hand, for efficiency,
it is important to prune the search through simulations. To use
$\Gg$-simulation, we need a finite $\Gg$-map. Since we are interested
in clock bounded reachability, we can inject the additional guard
$\Land_{x} x \le B$ in all transitions. The following theorem says
that for such automata, the reduced $\Gg$-map will be finite. This is
not true with non-reduced $\Gg$-maps. For instance, consider a
modification of the automaton in
Figure~\ref{fig:non-terminating-G-computation} with all transitions
having $x \le 3 \land y \le 3$. This does not help cutting any of the
three sources of infinite propagation that have been discussed in the
text below the figure.

\begin{theorem}
  \label{thm:clock-bounded}
  Suppose every transition of a UTA has a guard containing an upper
  constraint $x \lleq c$ for every clock. The reduced $\Gg$-map for
  such a UTA is finite.
\end{theorem}
\begin{proof}
  Let $M$ be the maximum constant appearing in a guard of the UTA
  (note that we have not considered the constants in the updates). Let
  $\Ggbar$ be the set of all atomic constraints with constant at most
  $M$. We will show that $\wp(\Ggbar, g, up) \incl \Ggbar$ for every
  transition $(q, g, up, q')$ and thus the $\Gg$-map associating
  $\Gg(q) = \Ggbar$ is a reduced $\Gg$-map. The proof proceeds
  similarly to
  Theorem~\ref{thm:reduced-g-bounds-finite-syntactically-bdd},
  although here there is an added convenience that the guard $g$
  contains an upper constraint for every clock. Suppose $g$ contains
  $x \lleq c$. Then, when $\upinv(\varphi)$ is an upper constraint
  $x \lleq d$, $\wp(\varphi, g, up)$ is $\top$ (Case 1); when
  $\varphi$ is $d \lleq x$, $\wp(\varphi, g, up)$ gives a constraint
  with constant atmost $c$; when $\varphi$ is a diagonal constraint
  $x - y \lleq d$ or $d \lleq x - y$ and Case 3 does not apply then
  $d \le c$ and hence $\wp(\varphi, g, up)$ is a constraint in
  $\Ggbar$.
\end{proof}

\subparagraph
{UTA with finite non-reduced $\Gg$-maps.} 
Given a finite set of atomic constraints $G$, the algorithm for
$Z \lug Z'$ first divides $G$ as $G^{nd} \cup G^d$ where $G^{nd}$
and $G^d$ are respectively the subsets of non-diagonal and diagonal
constraints in $G$. From $G^{nd}$, two functions
$L\colon X \mapsto \Nat \cup \{ -\infty\}$ and
$U\colon X \mapsto \Nat \cup \{-\infty\}$ are defined:
$L(x) = \max\{c \mid c \lleq x \in G^{nd} \}$ and
$U(x) = \max\{ c \mid x \lleq c \in G^{nd} \}$. When there is no
$c \lleq x$, $L(x) = -\infty$. Similarly for $U(x)$. Denote these
functions as $L(G)$ and $U(G)$. Once $G$ is rewritten as
$L(G), U(G)$ and $G^{d}$, \cite{CAV-19} gives a procedure to compute
$Z \lug Z'$.

For two $\Gg$-maps $\Gg_1$ and $\Gg_2$ we write
$LU(\Gg_2) \le LU(\Gg_1)$ if for every $q$ and every clock $x$,
$L(\Gg_2(q))(x) \le L(\Gg_1(q))(x)$ and
$U(\Gg_2(q))(x) \le U(\Gg_1(q))(x)$. We write $\Gg_2^d \incl \Gg_1^d$
if $\Gg_2(q)^d \incl \Gg_1(q)^d$ for every $q$.  It can be shown that
for two $\Gg$-maps $\Gg_1$ and $\Gg_2$ with $LU(\Gg_2) \le LU(\Gg_1)$
and $\Gg^d_2 \incl \Gg^d_1$, the $\Gg_2$-simulation is bigger than the
$\Gg_1$-simulation (using the observation made in the proof of
Proposition~\ref{prop:pre-is-sound} for non-diagonals and the direct
Definition~\ref{def:g-simulation} for diagonals). The following
theorem asserts that when the non-reduced $\Gg$-map is finite,
the reduced $\Gg$-map is finite and it induces a
bigger simulation. The proof of this theorem proceeds by showing
that every upper constraint $x \lleq c$ and diagonal constraint added
by the reduced propagation is also added by the non-reduced
propagation, and for every lower constraint $c \lleq x$ in the reduced
$\Gg$, there is some $c' \lleq' x$ in the non-reduced $\Gg$ with
$c \le c'$.

\begin{theorem}\label{thm:sound-finite-reduced-finite}
  When the smallest (non-reduced) $\Gg$-map $\Gg_1$ is finite, the
  smallest reduced $\Gg$-map $\Gg_2$ is also finite. Moreover,
  $LU(\Gg_2) \le LU(\Gg_1)$ and $\Gg_2^d \incl \Gg_1^d$.
\end{theorem}
\begin{proof}
  Let $\Gg^i_1(q)$ and $\Gg^i_2(q)$ be the $\Gg$-set at $q$ during the
  $i^{th}$ iteration of the respective fixed-point computations. We
  will show the following invariants at each iteration of the
  fixed-point computation:
  \begin{itemize}[nosep]
  \item every upper constraint $x \lleq c \in \Gg^i_2(q)$ is also
    present in $\Gg^i_1(q)$,
  \item for every lower constraint $c \lleq x \in \Gg^i_2(q)$, there
    is some $c' \lleq' x \in \Gg^i_1(q)$ with $c\leq c'$,
  \item every diagonal constraint in $\Gg^i_2(q)$ is present in
    $\Gg^i_1(q)$.
  \end{itemize}
  For upper constraints and diagonal constraints, whenever
  $\wp(\varphi, g, up)$ used in the $\Gg_2$ propagation gives a
  non-trivial constraint, it equals $\upinv(\varphi)$. This is the
  constraint used in the $\Gg_1$ computation. This shows the first and
  third invariants. We focus on the lower constraint. At the initial
  step $0$, in addition to the lower constraints coming from guards,
  $\Gg_1(q)$ contains $\upinv(0 \le x)$ for every outgoing transition
  $(q, g, up, q')$ and every clock. The reduced computation $\Gg_2(q)$
  contains all lower constraints present in guards, and
  $\wp(0 \le x, g, up)$. When $\wp(0 \le x, g, up) = \upinv(0 \le x)$,
  this constraint is present in $\Gg_1(q)$ as well. If
  $\wp(0 \le x, g, up) \neq \upinv(0 \le x)$, this is due to Case 2 of
  Table~\ref{tab:optimizations-main}. Hence there is a guard
  $x \lleq c$ in $g$, and $\upinv(0 \le x) = d \le x$ with $c<d$,
  whereas $\wp(0 \le x, g, up) = c\leq x$.  This shows the invariant.
  For the induction step, assume the invariant at $i$.  Pick a
  constraint $c \lleq x$ in $\Gg^i_2(q')$.  By hypothesis, there is
  $c' \lleq' x$ with $c \le c'$ in $\Gg^i_1(q')$.  Notice that
  $\upinv(c \lleq x)$ is either trivial, or is of the form
  $c - \e \lleq z$ when $up_x= z + \e$.  Clearly,
  $c - \e \le c' - \e$.  Thus, when
  $\wp(c \lleq x, g, up) = \upinv(c \lleq x)\in\Gg_2^{i+1}(q)$, then
  $c'-\e\lleq' z\in\Gg_1^{i+1}(q)$ and the invariant holds.  Otherwise
  $\wp(c \lleq x, g, up) = c_1 \leq z$ due to the presence of a guard
  $z \lleq_1 c_1$ in $g$ with $c_1< c - \e\leq c'-\e$.  Hence
  $c_1 \leq z \in \Gg^{i+1}_2(q)$ and
  $c' - \e \lleq' z \in \Gg^{i+1}_2(q)$.  Since this property is true
  for each lower constraint propagation, we get the invariant.
\end{proof}


\section{Termination of the reduced propagation}
\label{sec:termination}

We present an algorithm and discuss the complexity for the problem of
deciding whether the smallest reduced $\Gg$-map of a given automaton
is finite.  Briefly, we present a large enough bound $B$ such that if
the fixed point iteration does not terminate in $B$ steps, it will
never terminate and hence the smallest reduced $\Gg$-map given by the
least fixed-point is infinite.

Let us first assume that there are no strict inequalities in the
atomic constraints present in guards. For the termination analysis, we
can convert all strict inequalities $<$ to weak inequalities
$\le$. The reduced propagation does not modify the nature of the
inequality except in Case 2 of Table~\ref{tab:optimizations-main}
where strict may change to weak.  Any propagation in the original
automaton is preserved in the converted automaton with the same
constants and vice-versa. Hence the $\Gg$-map computation terminates
in one iff it terminates in the other. We denote by $c_\varphi$ the
constant of an atomic constraint $\varphi$.

Let $\Aa = (Q, X, q_0, T, F)$ be some UTA. Let
$M=\max\{c \mid c \text{ occurs in some guard of } \Aa\}$ and
$L=\max\{|d| \mid d \text{ occurs in some update of } \Aa\}$. Let
$\Gg$ be the smallest reduced $\Gg$-map computed by the least
fixed-point of the equations in
Lemma~\ref{lem:reduced-smallest-is-least-fixed-point}. We can show
that this fixed-point computation does not terminate iff a constraint
with a large constant is added to some $\Gg(q)$.

\begin{proposition}\label{prop:termination-large-constant}
  The reduced $\Gg$-map computation does not terminate iff for some
  state $q$, there is an atomic constraint $\varphi\in\Gg(q)$ with a
  constant $c_\varphi>N=\max(M,L)+2L|Q||X|^{2}$.
\end{proposition}

For the analysis, we make use of strings of the form $x \le$, $\le x$,
$x -y \le$ and $\le x - y$ where $x,y \in X$ and call them
\emph{contexts}.  Given a context $\context{\varphi}$ and a constant
$c$, we denote by $\context{\varphi}[c]$ the atomic constraint
obtained by plugging the constant into the context.

In the proof, we shall use the notion of propagation sequence, which
is a sequence
$(q_i,\context{\varphi}_i[c_i]) \to
(q_{i+1},\context{\varphi}_{i+1}[c_{i+1}]) \to \cdots \to
(q_j,\context{\varphi}_j[c_j])$ such that for all $i\leq k<j$ we have
$\context{\varphi}_{k+1}[c_{k+1}]=\wp(\context{\varphi}_k[c_k],g_k,up_k)$
for some transition $(q_{k+1},g_k,up_k,q_k)$ of $\Aa$.

\begin{proof}[Proof of Proposition~\ref{prop:termination-large-constant}.]
  The left to right implication of
  Proposition~\ref{prop:termination-large-constant} is clear.
  Conversely, assume that $\context{\varphi}[c]\in\Gg(q)$ for some
  $(q,\context{\varphi}[c])$ with $c>\max(M,L)+2L|Q||X|^{2}$.
  Consider the smallest $n\geq0$ such that
  $\context{\varphi}[c]\in\Gg^{n}(q)$.  There is a propagation
  sequence
  $\pi= (q_0,\context{\varphi}_0[c_0]) \to
  (q_1,\context{\varphi}_1[c_1]) \to \cdots \to
  (q_n,\context{\varphi}_n[c_n])$ such that
  $\context{\varphi}_0[c_0]\in\Gg^{0}(q_0)$ and
  $(q_n,\context{\varphi}_n[c_n])=(q,\context{\varphi}[c])$.  Notice
  that $\context{\varphi}_i[c_i]\in\Gg^{i}(q_i)$ for all
  $0\leq i\leq n$.  We first show that the propagation sequence $\pi$
  contains a positive cycle with large constants.
  
\begin{lemma}\label{lem:large-cycle}
  We can find $0<i<j\leq n$ such that
  $(q_i,\context{\varphi}_i)=(q_j,\context{\varphi}_j)$, $c_i<c_j$ and
  $\max(M,L)<c_k$ for all $i\leq k\leq j$.
\end{lemma}

\begin{proof}
  First, since $\context{\varphi}_0[c_0]\in\Gg^{0}(q_0)$ we have
  $0\leq c_0\leq\max(M,L)$.  We consider the last occurrence of a
  small constant in the propagation sequence.  More precisely, we
  define $m=\max\{k \mid 0\leq k<n \wedge c_k\leq\max(M,L)\}$. Hence,
  $c_k>\max(M,L)$ for all $m<k\leq n$.
  
  Notice that, for $m<k<n$, the constraint in the sequence cannot
  switch from an upper diagonal to a lower diagonal and vice-versa.
  Indeed, assume that $\context{\varphi}_k[c_k]=(x-y\leq c_k)$ and
  $\context{\varphi}_{k+1}[c_{k+1}]=(c_{k+1}\leq y'-x')$.  Then the
  update $up_k$ contains $x:=x'+d$, $y:=y'-e$ with $c_{k+1}=d+e-c_k$.
  This is a contradiction with $d,e\leq L$ and $c_k,c_{k+1}>L$.
  Similarly, we can show that an upper (resp.\ lower) diagonal
  constraint cannot switch to a lower (resp.\ upper) non-diagonal
  constraint.  On the other hand, it is possible to switch once from
  an upper (resp.\ lower) diagonal constraint to an upper (resp.\
  lower) non-diagonal constraint.
  
  The other remark is that $|c_{k+1}-c_k|\leq2L$ for all $m\leq k<n$.
  Since $c_m\leq\max(M,L)$ and $c_n>\max(M,L)+2L|Q||X|^{2}$, we find
  an increasing sequence $m<i_1<i_2<\cdots<i_\ell\leq n$ with
  $c_{i_1}<c_{i_2}<\cdots<c_{i_\ell}$ and $\ell>|Q||X|^{2}$.  As
  noticed above, the $\context{\varphi}_k$ are either all upper
  constraints or all lower constraints, hence the set
  $\{(q_k,\context{\varphi}_k) \mid m<k\leq n\}$ contains at most
  $|Q||X|^{2}$ elements ($|X|$ for non-diagonals and $|X|(|X|-1)$ for
  diagonals). Therefore, we find $i,j\in\{i_1,\ldots,i_\ell\}$ such
  that $i<j$ and
  $(q_i,\context{\varphi}_i)=(q_j,\context{\varphi}_j)$. Recall that
  $c_k>\max(M,L)$ for all $m<k\leq n$.
\end{proof}

The next step is to show that a positive cycle with large constants
can be iterated resulting in larger and larger constants.

\begin{lemma}\label{lem:large-cycle-iteration}
  Let
  $(q_i,\context{\varphi}_i[c_i]) \to
  (q_{i+1},\context{\varphi}_{i+1}[c_{i+1}]) \to \cdots \to
  (q_j,\context{\varphi}_j[c_j])$ be a propagation sequence with
  $(q_i,\context{\varphi}_i)=(q_j,\context{\varphi}_j)$,
  $\delta=c_j-c_i>0$ and $M<c_k$ for all $i\leq k\leq j$.  Then,
  $(q_i,\context{\varphi}_i[c_i+\delta]) \to
  (q_{i+1},\context{\varphi}_{i+1}[c_{i+1}+\delta]) \to \cdots \to
  (q_j,\context{\varphi}_j[c_j+\delta])$ is also a propagation
  sequence.
\end{lemma}

\begin{proof}
  Let $i\leq k<j$ and $(q_{k+1},g_k,up_k,q_k)$ be a transition of
  $\Aa$ yielding the propagation from $k$ to $k+1$:
  $\context{\varphi}_{k+1}[c_{k+1}]=\wp(\context{\varphi}_k[c_k],g_k,up_k)$.
  Since $c_{k+1}>M$, none of the three cases of
  Table~\ref{tab:optimizations-main} applies: if
  $\context{\varphi}_{k+1}$ is one of $x\leq$, $\leq x$, $x-y\leq$ or
  $\leq x-y$ then $g_k$ does not contain $x\leq c$ or $x-y\leq
  c$. Hence, we have
  $\context{\varphi}_{k+1}[c_{k+1}]=up^{-1}(\context{\varphi}_k[c_k])$.
  We deduce that
  $\context{\varphi}_{k+1}[c_{k+1}+\delta]=up^{-1}(\context{\varphi}_k[c_k+\delta])$.
  Since $c_{k+1}+\delta>M$ the cases of
  Table~\ref{tab:optimizations-main} do not apply and we get
  $\context{\varphi}_{k+1}[c_{k+1}+\delta]=\wp(\context{\varphi}_k[c_k+\delta],g_k,up_k)$.
\end{proof}

This allows to conclude the proof of
Proposition~\ref{prop:termination-large-constant}.  Using
Lemma~\ref{lem:large-cycle} we obtain from $\pi$ a positive cycle with
large constants. This cycle can be iterated forever thanks to
Lemma~\ref{lem:large-cycle-iteration}. We deduce that
$\context{\varphi}_i[c_i+k\delta]\in\Gg^{i}(q_i)$ for all $k\geq0$ and
the reduced $\Gg$-computation does not terminate.
\end{proof}

\subparagraph*{Algorithm to detect
  termination.}\label{termination-algo}
Proposition~\ref{prop:termination-large-constant} gives a termination
mechanism: run the fixed-point computation $\Gg^{0},\Gg^{1},\ldots$,
stop if either it stabilises with $\Gg^{n}=\Gg^{n+1}$ or if we add
some constraint $\varphi\in\Gg^{n}(q)$ with $c_\varphi>N$.  The number
of pairs $(q,\varphi)$ with $c_\varphi\leq N$ is $2N|Q||X|^{2}$ (the
factor 2 is for upper or lower constraints).  Therefore, the
fixed-point computation stops after at most $2N|Q||X|^{2}$ steps and
the total computation time is $\mathsf{poly}(M,L,|Q|,|X|)$.  If the
constants occurring in guards and updates of the UTA $\Aa$ are encoded
in unary, the static analysis terminates in time
$\mathsf{poly}(|\Aa|)$.  If the constants are encoded in binary,
(non-)termination of the $\Gg$-computation can be detected in
$\textsc{NPspace}=\PSPACE$: it suffices to search for a propagation
sequence
$(q_0,\varphi_0)\to(q_1,{\varphi}_1)\to\cdots\to(q_n,{\varphi}_n)$
such that ${\varphi}_0\in\Gg^{0}(q_0)$ and $c_{\varphi_n}>N$.  For
this, we only need to store the current pair $(q_k,\varphi_k)$, guess
some transition $(q_{k+1},g_k,up_k,q_k)$ of $\Aa$, and compute the
next pair $(q_{k+1},\varphi_{k+1})$ with
$\varphi_{k+1}=\wp(\varphi_k,g_k,up_k)$.  This can be done with
polynomial space.  We can also show a matching $\PSPACE$ lower-bound.

\subparagraph*{Lower bound.}  We now show that when constants are
encoded in binary, deciding termination of the reduced propagation is
$\PSPACE$-hard. To do this, we give a reduction from the control-state
reachability of bounded one-counter automata.

A \emph{bounded one-counter
  automaton}~\cite{Haase:FundInf-RP:2016,FearnleyJ15:2015:IandC} is
given by $(L, \ell_0, \Delta, b)$ where $L$ is a finite set of states,
$\ell_0$ is an initial state, $\Delta$ is a set of transitions and
$b \ge 0$ is the global bound for the counter. Each transition is of
the form $(\ell, p, \ell')$ where $\ell$ is the source and $\ell'$ the
target state of the transition, $p \in [-b, +b]$ gives the update to
the counter. A run of the counter automaton is a sequence
$(\ell_0, c_0) \to (\ell_1, c_1) \to \cdots \to (\ell_n, c_n)$ such
that $c_0 = 0$, each $c_i \in [0, b]$ and there are transitions
$(\ell_i, p_i, \ell_{i+1})$ with $c_{i+1} = c_i + p_i$. All constants
used in the automaton definition are encoded in binary. Reachability
problem for this model asks if there exists a run starting from
$(\ell_0, 0)$ to a given state $\ell_t$ with any counter value
$c_t$. This problem is known to be
$\PSPACE$-complete~\cite{FearnleyJ15:2015:IandC}. We will now reduce
the reachability for bounded one-counter automata to the problem of
checking if the fixed-point computing the smallest reduced $\Gg$-map
terminates (i.e, whether the smallest reduced $\Gg$-map is finite).

From a bounded one counter automaton $\Bb = (L, \ell_0, \Delta, b)$ we
construct a UTA $\Aa_{\Bb}$. States of $\Aa_\Bb$ are
$L \cup \{\ell'_0, \ell'_t\}$ where $\ell'_0$ and $\ell'_t$ are new
states not in $L$. There are two clocks $x, y$. For each transition
$(\ell, p, \ell')$ of $\Bb$, there is a transition
$(\ell', g, up, \ell)$ with guard $x \le b \land y \le 0$ and updates
$x:= x - p$ and $y:= y$. We add some extra transitions using the new
states $\ell'_0$ and $\ell'_t$: (1)
$\ell_0 \xra{x - y \le 0} \ell'_0$, (2) $\ell'_t \xra{} \ell_t$ and
(3) $\ell'_t \xra{~x:= x, y:= y + 1~} \ell'_t$.

Consider Case 3 of Table~\ref{tab:optimizations-main}. A guard of the
form $x \le b$ disallows propagation of constraints $x - y \le d$ with
$ d > b$. But, it can allow $d$ to go smaller and smaller, and at one
point the constant becomes negative and the constraint gets rewritten:
$x - y \le b, x - y \le b - 1, \dots, x - y \le 0, 1 \le y - x, 2 \le
y -x$, etc. The presence of a constraint $y \le 0$ will eliminate
$1 \le y - x, 2 \le y -x$, etc. once again due to Case 3. We make use
of this facility to simulate a bounded one-counter automaton.

\begin{lemma}\label{lem:counter-to-timed}
  For every run
  $(\ell_0, 0) \to (\ell_1, c_1) \to \cdots \to (\ell_n, c_n)$ in
  $\Bb$, there is a propagation sequence
  $(\ell_0, x - y \le 0) \to (\ell_1, x - y \le c_1) \to \cdots \to
  (\ell_n, x - y \le c_n)$ in $\Aa_{\Bb}$.
\end{lemma}
\begin{proof}
  Let $\Gg$ be the smallest reduced $\Gg$-map of $\Aa_{\Bb}$. We will
  show by induction that for every $i$, there is a constraint
  $x - y \le c_i$ in $\Gg(\ell_i)$.  Due to the edge
  $\ell_0 \xra{x - y \le 0} \ell'_0$, we have
  $x - y \le 0 \in \Gg(\ell_0)$. Suppose the hypothesis is true for
  some $j > 0$, that is, we have $x - y \le c_j \in
  \Gg(\ell_j)$. Since we have
  $(\ell_j, c_j) \to (\ell_{j+1}, c_{j+1})$ there is a transition
  $(\ell_j, p, \ell_{j+1})$ in $\Bb$, and $c_{j+1} = c_j + p$ with
  $0 \le c_{j+1} \le b$. By construction, there is a transition
  $(\ell_{j+1}, g, up, \ell_j)$ in $\Aa_{\Bb}$ with $up_x = x - p$ and
  guard $x \le b \land y \le 0$. Hence the constraint
  $\varphi = x - y \le c_j$ at $\ell_j$ should be propagated to
  $\ell_{j+1}$. We have $\upinv(\varphi) = x - y \le c_j + p$, that is
  $x - y \le c_{j+1}$. Since $0 \le c_{j+1} \le b$, Case 3 of
  Table~\ref{tab:optimizations-main} does not apply: the constraint
  $x\leq b$ in $g$ does not cut the propagation since $c_{j+1}\leq b$,
  and the constraint $y\leq0$ in $g$ does not apply as well (if
  $c_{j+1}=0$ the constraint $up^{-1}(\varphi)$ can also be written as
  $0\leq y-x$).  Therefore,
  $\wp(x - y \le c_{j}, g, up) = x - y \le c_{j+1} \in
  \Gg(\ell_{j+1})$.
\end{proof}

\begin{lemma}\label{lem:timed-to-counter}
  For every propagation sequence
  $(\ell_0, x - y \le 0) \to (\ell_1, x - y \le c_1) \to \cdots \to
  (\ell_n, x - y \le c_n)$ in $\Aa_{\Bb}$ with $\ell_i \in L$ for
  $0 \le i \le n$, there is a run
  $(\ell_0, 0) \to (\ell_1, c_1) \to \cdots \to (\ell_n, c_n)$ in
  $\Bb$.
\end{lemma}
\begin{proof}
  The proof is by induction on $n$.  The base case $n = 0$ is trivial.
  Let $n > 0$ and suppose the lemma is true up to $n-1$.  Consider the
  propagation
  $(\ell_{n-1}, x - y \le c_{n-1}) \to (\ell_{n}, x -y \le c_{n})$.
  This implies there is a transition $(\ell_{n}, g, up, \ell_{n-1})$
  in $\Aa_{\Bb}$ with $up_x = x - p$ for $p=c_{n} - c_{n-1}$ and
  $\wp(x - y \le c_{n-1}, g , up) = x - y \le c_{n}$.  By
  construction, every $g$ is $x \le b \land y \le 0$.  As
  $\wp(x - y \le c_{n-1}, g , up)$ is non-trivial, firstly
  $c_{n} \le b$ (otherwise Case 3 of
  Table~\ref{tab:optimizations-main} will apply) and secondly
  $0 \le c_{n}$.  To see this, suppose $c_{n} < 0$, the constraint
  $x - y \le c_{n}$ gets rewritten as $-c_{n} \le y - x$ and the
  propagation would give $(\ell_{n},-c_{n}\leq y-x)$ contrary to what
  was assumed.  Now, we consider the counter automaton $\Bb$.  By
  induction hypothesis, there is a run up to $(\ell_{n-1}, c_{n-1})$.
  From the transition $(\ell_{n}, g, up, \ell_{n-1})$ of $\Aa_{\Bb}$,
  we know there is a transition $(\ell_{n-1}, p , \ell_{n})$.  We have
  seen that $0 \le c_{n} \le b$.  Hence there is a step
  $(\ell_{n-1}, c_{n-1}) \to (\ell_{n}, c_{n})$ in $\Bb$, giving an
  extension to the run.
\end{proof}

\begin{proposition}\label{thm:diagonal-pspace-hard}
  The final state is reachable in the counter automaton $\Bb$ iff the
  smallest reduced $\Gg$-map of $\Aa_{\Bb}$ is infinite.
\end{proposition}
\begin{proof}
  Let $\Gg$ be the smallest reduced $\Gg$-map of $\Aa_{\Bb}$.

  Suppose the final state is reachable in $\Bb$, with a run
  $(\ell_0, c_0) \to \cdots \to (\ell_t, c_t)$. From
  Lemma~\ref{lem:counter-to-timed}, there is a propagation giving
  $x - y \le c_t$ in $\Gg(\ell_t)$. Due to the extra transitions
  $\ell'_t \to \ell_t$ and $\ell'_t \xra{x:= x, y:=y + 1} \ell'_t$ in
  $\Aa_{\Bb}$ (with no guards), we deduce that
  $x - y \le c_t + i \in \Gg(\ell'_t)$ for all $i \ge 0$.

  Suppose $\Gg$ is infinite. There is some state $\ell$ with an
  infinite $\Gg(\ell)$. Firstly $\ell \neq \ell'_0$, since there are
  no outgoing transitions from $\ell'_0$ and therefore
  $\Gg(\ell'_0) = \emptyset$. Pick some state
  $\ell \in L \setminus \{\ell'_0, \ell'_t\}$. Every outgoing
  transition from $\ell$ has guard $x \le b \land y \le 0$, except in
  the case $\ell_0 \xra{x - y \le 0} \ell'_0$. But since
  $\Gg(\ell'_0) = \emptyset$, we can forget this transition as far as
  propagation is concerned.  Hence Table~\ref{tab:optimizations-main}
  ensures that we have $d \le b$ for every atomic constraint
  $x - y \leq d$ or $d \leq x - y$ propagated to $\Gg(\ell)$.
  Moreover, if at all there is a propagated constraint $d \leq y - x$
  or $y - x \leq d$, then $d = 0$ since $y \le 0$ is present in all
  outgoing guards.  This shows that the number of diagonal constraints
  is finite in $\Gg(\ell)$.  In the construction, a diagonal
  constraint $x - y \le c$ always propagates as a diagonal since all
  updates are of the form $x := x - p$ and $y:= y$.  Coming to
  non-diagonals. Since we have $x \le b$ and $y\leq0$ in all guards of
  outgoing transitions from $\ell$, Case 1 disallows propagation of
  any other upper constraint to $\ell$, and Case 2 bounds the possible
  constants of lower constraints $d\leq x$ or $d\leq y$ in
  $\Gg(\ell)$.  This gives finite $\Gg(\ell)$ for
  $\ell \in L \cup \{\ell'_0\}$. Therefore, the only possibility is to
  have $\Gg(\ell'_t)$ infinite, due to the self-loop on $\ell'_t$ with
  update $up$ being $x:= x$ and $y:= y + 1$. The infinite number of
  constraints arises due to $y := y + 1$ and hence should come from a
  constraint that involves $y$. Since there are no guards in this
  self-loop, $\wp(\varphi, \top, up) = \upinv(\varphi)$.  For any
  constraint $\context{\varphi}[c]$, we have
  $\upinv(\context{\varphi}[c]) = \context{\varphi}[c - 1]$ if $y$
  occurs with a positive sign in $\varphi$ and
  $\context{\varphi}[c+1]$ if $y$ occurs with a negative sign in
  $\varphi$. There are three ways a constraint involving $y$ reaches
  $\ell'_t$ during the propagation. One of them is $y \le 0$ which
  could come from $\ell_t$. But the pre of this is $y \le -1$ and
  hence is trivial. Another possibility is from $\wp(0\leq y,\top,up)$
  used in the initialization step $\Gg^{0}$. But this gives $-1\leq y$
  which is trivial. The only other way to have a propagation is to
  start from $x - y \le 0 \in \Gg(\ell_0)$, reach some
  $x - y \le c_t \in \Gg(\ell_t)$ with $0 \le c_t \le b$. This then
  passes on to $\Gg(\ell'_t)$. Starting from this, we get constraints
  $x - y \le c_t + i$ for $i \ge 0$ in $\Gg(\ell'_t)$. This gives a
  propagation sequence
  $(\ell_0, x - y \le 0) \to \cdots (\ell_t, x - y \le c_t)$. From
  Lemma~\ref{lem:timed-to-counter}, there is a corresponding run in
  the counter automaton $\Bb$.
\end{proof}

\begin{theorem}\label{thm:termination}
  Deciding termination of the reduced $\Gg$-map computation for a
  given UTA $\Aa$ is in $\textsc{Ptime}$ if the constants in $\Aa$ are
  encoded in unary, and $\PSPACE$-complete if the constants are
  encoded in binary.
\end{theorem}
\begin{proof}
  The algorithm to detect termination given in
  Page~\pageref{termination-algo} discusses the upper bound,
  $\textsc{Ptime}$ when constants are in unary and $\PSPACE$ when the
  constants are in binary.  The $\PSPACE$ lower-bound is proved in
  Proposition~\ref{thm:diagonal-pspace-hard}, hence the problem is
  $\PSPACE$-complete.
\end{proof}


\section{Experiments}
\label{sec:experiments}

\begin{table}[t]
  \centering
  \begin{tabular}{|l|c|c|c|c|c|}
    \hline
    & & \multicolumn{2}{c|}{New static analysis} &
                                                   \multicolumn{2}{c|}{Static analysis
                                                   of~\cite{CAV-19}}
    \\
    \hline
    Model & Schedulable? & \# nodes & time & \# nodes  & time \\
    \hline

    SporadicPeriodic-5 & Yes & 677 & 1.710s & - & - \\
    SporadicPeriodic-20 & No & 852 & 1.742s & - & - \\
    \hline
    Mine-Pump & Yes & 31352 & 7m 23.509s & - & -  \\
    \hline
    \multicolumn{6}{|l|}{\emph{Flower} task triggering automaton: (computation
    time, deadline)}   \\
    (1,2), (1,2), (1,2) & No & 212 & 0.057s & - & - \\
    (1,10), (1,10), (1,10), (1,4) & Yes & 105242 & 8m 57.256s & -
                         & - \\
    \hline
    \multicolumn{6}{|l|}{\emph{Worst-case} task triggering automaton: (computation
    time, deadline)}   \\
    (1,2), (1,2), (1,2) & No & 20  & 0.050s & - & - \\
    (1,10), (1,10), (1,10), (1,4) & Yes & 429 &
                                                0.454s
          & - & - \\
    12 copies of (1,20) & Yes & 786 & 12m 5.250s & - & - \\
    \hline
    $\Aa_{gain} \times 3$ & \textsf{N/A} & 24389 & 7.611s & 24389 & 12.402s \\
    $\Aa_{gain} \times 4$ & \textsf{N/A} & 707281 & 14m 12.369s & 707281 & 27m 13.540s  \\
    \hline
  \end{tabular}
  \caption{
    \# nodes is the number of nodes enumerated during a
    breadth-first-search; ``-'' denote that there was no answer for 20
    minutes; \textsf{N/A} denotes not-applicable, as $\Aa_{gain}$ is
    not a scheduling problem.}
  \label{tab:experiments-new}
\end{table}

We report on experiments conducted using the open source tool
TChecker~\cite{tchecker}. The models are given as networks of timed
automata which communicate via synchronized actions. We have
implemented the new static analysis discussed in
Section~\ref{sec:better-gg-bounds}. The older static analysis and zone
enumeration with the $\Gg$-simulation were already
implemented~\cite{CAV-19}.

\subparagraph
{Compositionality of static analysis.} Both these
static analyses are performed individually on each component. For each
local state $q_i$ a map $\Gg(q_i)$ is computed. During the zone
enumeration the product of the automata is computed on-the-fly. Each
node is of the form $(q, Z)$ where
$q = \langle q_1, q_2, \dots, q_k \rangle$ is a tuple of local states,
one from each component of the network and $Z$ is a zone over all
clocks of the network. The $\Gg$-map is then taken as
$\Gg(q) = \bigcup_i \Gg(q_i)$. This approach creates a problem when
there are shared clocks. A component $i$ might update $x$ and another
component $j \neq i$ might contain a guard with $x$. The $\Gg$-maps
computed component-wise will then not give a sound simulation. In our
experiments, we construct models without shared clocks.

\subparagraph
{Benchmarks.} Our primary benchmarks are
models of task scheduling problems using the Earliest-Deadline-First
(EDF) policy. Each task has a computation time and a deadline. Tasks
are released either periodically or via a specification given as a
timed automaton. The goal is to verify if for a given set of released
tasks, all of them can be finished within their deadline. Preemption
of tasks is allowed. This problem has been encoded as a reachability
in a network of timed automata that uses bounded
subtraction~\cite{Fersman:InfComp:2007}.  The main challenge is to
model preemption. Each task $t_i$ has an associated clock $c_i$ which
is reset as soon as the task starts to execute. While $t_i$ is
running, and some other task $t_j$ preempts $t_i$, the clock $c_i$
continues to elapse time. When $t_j$ is done, an update
$c_i := c_i - C_j$ is performed, where $C_j$ denotes the computation
time of $t_j$. This way, when $t_i$ is scheduled again, clock $c_i$
maintains the computation time that has elapsed since it was
started. Whenever the EDF scheduler has to choose between task $t_i$
and $t_j$, it chooses the one which is closest to its deadline. To get
this, when $t_i$ is released, a clock $d_i$ is reset. Task $t_i$ is
prioritized over $t_j$ if $D_i - d_i < D_j - d_j$ where $D_i, D_j$ are
the deadlines. We have constructed a model for the EDF scheduler based
on these ideas, which we explain in more detail in
Section~\ref{sec:benchmarks}.  For the experiments in
Table~\ref{tab:experiments-new}, we consider some task release
strategies given in the literature (SporadicPeriodic from TIMES tool,
and a variant of Mine-Pump from \cite{GERDSMEIER2001143}) and also
create some of our own (Flower and Worst-case task triggers). he last
model $\Aa_{gain}$ is an automaton with reset-to-zero only updates
illustrating the gain when both static analyses terminate. More
details about the models are given in Section~\ref{sec:benchmarks}.

\subparagraph
{Comparison.} For all the EDF examples, the
old static analysis did not terminate, as seen in the last two columns
of Table~\ref{tab:experiments-new}. This is expected since the model
contains an update of the form $x := x - C$ which repeatedly adds
guards $x \le K, x \le K + C, x \le K + 2C, \dots$. The new static
analysis cuts this out, since the update $x := x - C$ occurs along
with a guard $x \le D$, making it a timed automaton with bounded
subtraction. The $\Aa_{gain}$ example runs with both the static
analyses. However, the new static analysis minimizes the propagation
of diagonal constraints. The time taken by the simulation test used in
the zone enumeration phase is highly sensitive to the number of
diagonal constraints. Fewer diagonals therefore result in a faster
zone enumeration.  We have also tried our new static analysis for
standard benchmarks of diagonal-free timed automata and observed no
gain. In these models, the distance between a clock reset and a
corresponding guard (in a component automaton) is
small, usually within one or two transitions. Hence resets already cut
out most of the guards and the optimizations of
Table~\ref{tab:optimizations-main} do not seem to help here. We expect
to gain primarily in the presence of updates or diagonal
constraints. We also remark that the last experiment cannot be
performed on the TIMES tool which is built for scheduling problems and
the previous ones cannot be modeled in other timed automata tools
UPPAAL, PAT and Theta since they cannot handle timed automata with
subtraction updates. Our prototype therefore subsumes existing tools
in terms of modeling capability.

\section{Benchmarks}
\label{sec:benchmarks}

We give a more detailed account of the models that we have considered
in our experiments. We first start with our EDF scheduling model, and
then describe the model for the last experiment in Table
\ref{tab:experiments-new}.

\subsection{Modelling EDF schedulability}

As mentioned in Section~\ref{sec:experiments}, the scheduling problem
is to know if a given set of tasks can all be scheduled within their
deadline through the EDF policy. We place a restriction on the task
modeling: at any point of time, there is only one instance of a
particular task. This restriction is not upfront
considered in \cite{Fersman:InfComp:2007} where multiple instances
of a task are allowed. However, there is a bound on the number of
instances. Hence different instances of a task can be renamed and
modeled in our setting. We model the EDF scheduling
problem using a network of timed automata consisting of a
\emph{scheduler}, a \emph{task handler} for each task, and a
\emph{task release} automaton.

\begin{figure}[h!]
  \centering
  \begin{tikzpicture}
    \node (0) at (0,0) {\includegraphics[height=200mm,
      width=135mm]{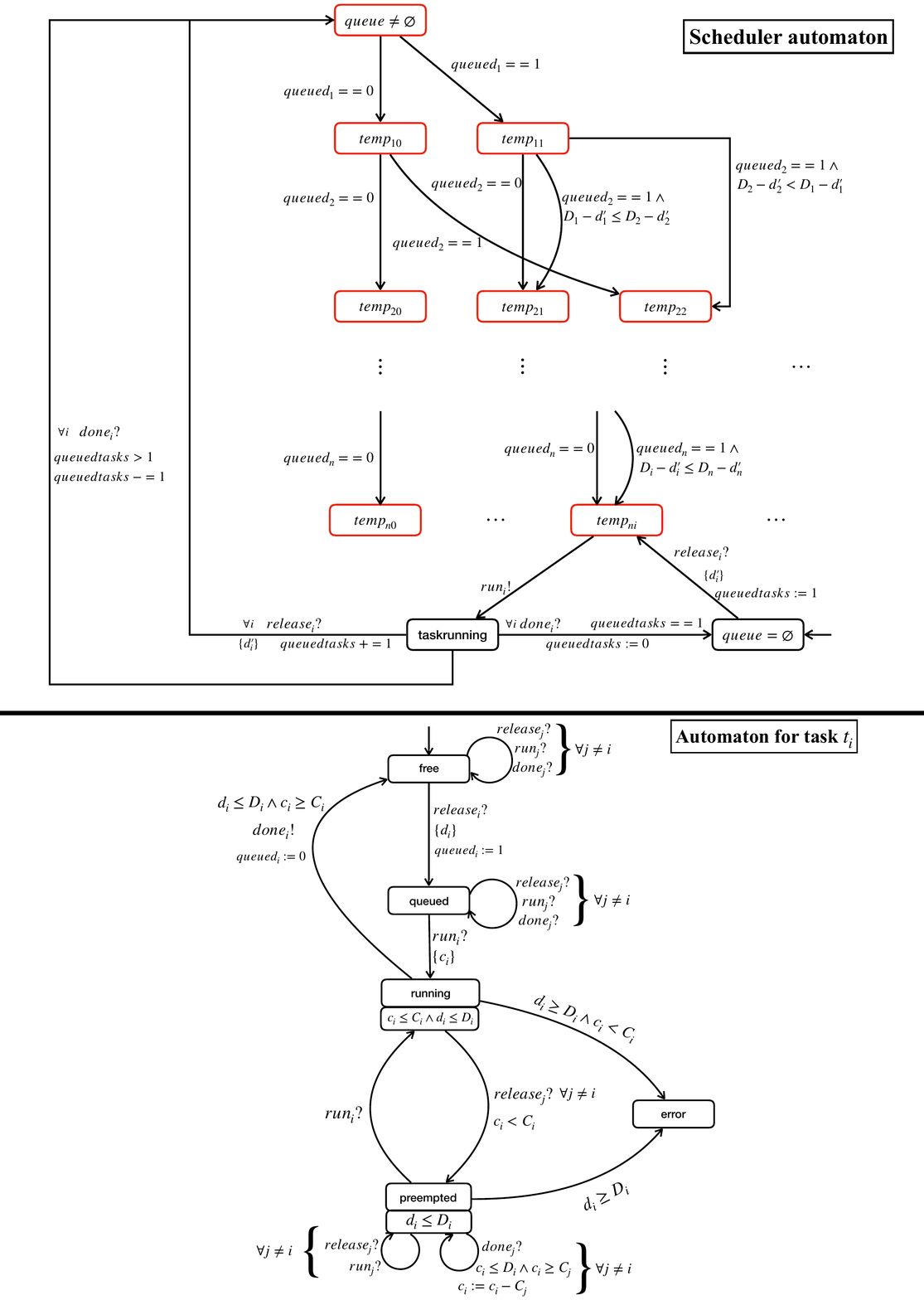}};
  \end{tikzpicture}
  \caption{EDF scheduler and task handler}
  \label{fig:edf-scheduler}
\end{figure}

\subparagraph
{Scheduler and task handler.}
Figure~\ref{fig:edf-scheduler} shows the scheduler and the task
handler automata. We use a boolean variable $queued_i$ to denote
whether task $t_i$ has been added to the queue ($queued_i = 1$) or not
($queued_i = 0$). The scheduler automaton selects which task gets
executed depending on the time left to reach its deadline. The
scheduler automaton has a state to remember that the task queue is
empty ($queue = \emptyset$). The state $taskrunning$ denotes that one
of the tasks (present in the queue) is being executed. The remaining
states, all marked red, constitute the gadget that chooses the task to
be executed. We mark a state red to denote that it is a
\emph{committed state}, meaning no time is allowed to elapse in this
state~\cite{Bengtsson:Springer:2004}. Assuming there are $n$ tasks, we
have $n$ \emph{layers} of these states. In the $i^{th}$ layer, there
are $(i+1)$ many states, $temp_{i0}, temp_{i1}, \dots, temp_{ii}$.  We
assume that the set of tasks is ordered. The state $temp_{ij}$ denotes
that among the first $i$ tasks, $t_1, t_2, \dots, t_i$ (some of these
tasks might not be queued), task $t_j$ (which must be queued) has the
closest deadline among the tasks that are queued. The state
$temp_{i0}$ denotes that none of the first $i$ tasks are queued. Note
that after checking the first $i$ tasks, if the task with the earliest
deadline is $t_j$, then after checking the $(i+1)^{th}$ task, the task
with the earliest deadline can either be $t_{(i+1)}$ or remain as
$t_j$.  The automaton thus has a transition from $temp_{ij}$ to
$temp_{(i+1)(i+1)}$ checking if the deadline of $t_{i+1}$ is
(strictly) closer than the deadline of $t_j$. This is checked by the
guard $D_{i+1}-d_{i+1}' < D_j - d_j'$, where $d_j'$ is a clock that
gets reset as soon as the task $t_j$ gets added to the queue (i.e. on
every transition with synchronization $release_j$). Otherwise, after
checking the $(i+1)^{th}$ task, $t_j$ remains to be the task with the
earliest deadline. This is possible in two scenarios - (i) the task
$t_{(i+1)}$ is not present in the queue ($queued_{i+1}=0$) or (ii) the
deadline of $t_{(i+1)}$ is atleast as far as the deadline of $t_j$
(checked by the guard $D_j - d_j' \le D_{i+1} - d_{i+1}'$).  There are
these two edges from $temp_{ij}$ to $temp_{(i+1)j}$, for every
$i = 1,2,\dots,n-1$ and $1 \le j \le i$.  The states $temp_{ni}$ (for
$i=1,2,\dots,n$) denote that among all the $n$ tasks, the task $t_i$
has the earliest deadline and hence this task must be executed. This
is ensured using the transition $temp_{ni} \to taskrunning$ with the
synchronization $run_i$.  This \emph{triangle-like} gadget chooses the
task with the earliest deadline.

While at the state $taskrunning$, if a task gets released, then we
again need to choose the task, among the queued tasks (including this
newly queued task), with the earliest deadline. For this, there is an
edge $taskrunning \to `queue\neq \emptyset$' with the synchronization
$release_i$, for every $i = 1,2,\dots,n$.  After a task gets finished,
there are two cases - either the queue becomes empty (the edges
$taskrunning \to `queue=\emptyset$') or the queue remains non-empty
(the edges $taskrunning \to `queue \neq \emptyset$'). These edges are
synchronized using the signal $done_i$ for every $i = 1,2,\dots,n$.

The task handler automaton for task $t_i$ maintains
the state of a task, whether the task has been added to the task queue
($queued$), whether it is running ($running$) or it has been preempted
($preempted$) or if the deadline of this task has been violated
($error$).  The state $free$ denotes that the task has not been
released yet.  As soon as the task gets released, we reset a clock
$d_i$ that tracks the deadline of the task.  This is done in the
transition $free \to queued$.  (This clock $d_i$ is essentially a copy
of the clock $d_i'$ present in the scheduler automaton.  We use these
copies to avoid using shared clocks.)  We use another clock $c_i$, to
maintain the execution time of the task.  This clock gets reset as
soon as the task starts to be executed, on the transition
$queued \to running$.  Since a task can be released while another task
is being executed, assuming there are $n$ many tasks, there are
$(n-1)$ many edges $running \to preempted$ on each of the signals
$release_j$ where $j \neq i$.  While a task is preempted, there can be
other tasks getting released, run or finished.  In order to account
for the cases for releasing and executing of other tasks, there is a
self loop on each of the signals $release_j$ and $run_j$, for
$j \neq i$. Since while the task $t_i$ is preempted, no other instance
of $t_i$ gets released into the queue, there is no self loop on either
$release_i$ or $run_i$. Note, while the automaton is in the state
$preempted$, the clock $c_i$ keeps elapsing time, although the task is
not being executed.  So, while the task $t_i$ is preempted, when
another task $t_j$ gets finished (synchronized using the signal
$done_j$), we deduct the computation time of $t_j$ from the clock
$c_i$. This is done using the update $c_i := c_i - C_j$ on the edge
$preempted \xrightarrow{done_j?} preempted$. In this edge, we also
have the guard $c_i \le D_i$.  This ensures that this automaton (for
task $t_i$) is a timed automaton with bounded subtraction. In the
state $running$, the invariant $c_i \le C_i \wedge d_i \le D_i$
ensures that the task is yet to finish its execution and also the
deadline has not been violated. The invariant at the state $preempted$
also ensures that, while at this state, the deadline of the task has
not been violated. There are two edges to the $error$ state, from
$running$ and from $preempted$, both checking that the deadline is
going to be (or has been) violated. The edge $running \to free$
ensures that the task has been executed within its deadline. The self
loops on the states $free$ and $queued$ ensure that there is no
deadlock when another task gets released or being executed or has been
done.

\subparagraph
{Task release automata.}
The release of a
task is controlled using a task release automaton. We have considered
four such automata for our experiments, corresponding to the first
four rows in Table~\ref{tab:experiments-new}. Figure~\ref{fig:flower}
and Figure~\ref{fig:worst-case-task} give the third and fourth rows:
\emph{Flower} and \emph{Worst-case} task triggering automata
respectively. Both of these automata ensure that a task $t_i$ gets
released only if it is not present in the task queue (i.e.
$queued_i = 0$). The first two rows of the table come
from existing task release strategies in the literature, which we
discuss below.

A \emph{periodic task} with a period $P$ is a task that gets released
every $P$ time units.  In order to maintain the fact that not more
than one instance of a task gets added to the task queue, the
deadlines for every periodic task, that we can consider, is lesser
than or equal to its period.  A set of periodic tasks is modeled using
the automaton in Figure~\ref{fig:periodic-tasks}.

The \emph{SporadicPeriodic} model is from the tool TIMES. This example
has three periodic tasks and one sporadic task, $A$, that is
controlled by the task releasing automaton given in
Figure~\ref{fig:sporadic-periodic}. We have run multiple examples for
different values of the constant $N$ present in the task releasing
automaton and for each of these models, the answer (of whether the
tasks are schedulable or not) matches with the answer of TIMES. The
\emph{Mine-Pump} example is presented in
\cite{GERDSMEIER2001143}. This example has six periodic tasks, written
as $(computation\ time, deadline, period)$, that are $(58,200,200)$,
$(37,250,250)$, $(37,300,300)$, $(39,350,350)$, $(33,800,800)$,
$(33,1000,1000)$. We only consider the first five tasks for our
experiment. In this example, our algorithm says that it is
schedulable, which matches with the result given in
\cite{GERDSMEIER2001143}.

\begin{figure}[t]
  \centering
  \begin{tikzpicture}[state/.style={circle, thick, draw, inner
      sep=2pt, minimum size = 3mm}]
    \begin{scope}[every node/.style={state}]
      \node (0) at (0,0) {\footnotesize $q_0$};
    \end{scope}
    \begin{scope}[->,>=stealth, thick]
      \draw (-1,0) to (0); \draw [rounded corners=5] (0) to (1,0.5) to
      (1,-0.5) to (0);
    \end{scope}
    \node at (2.2,0.3) {$queued_i==0$}; \node at (1.8,-0.3)
    {$release_i!$}; \node at (0.5,0.9) {$\vdots$}; \node at (0.5,-0.7)
    {$\vdots$};
  \end{tikzpicture}
  \caption{$\Aa_{flower}$: non-deterministically releases a task
    whenever the task is not already present in the queue}
  \label{fig:flower}
\end{figure}
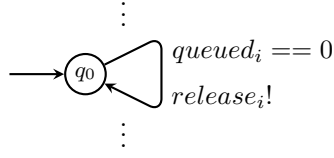

\begin{figure}[t]
  \centering
  \begin{tikzpicture}[state/.style={circle, thick, draw, inner
      sep=2pt, minimum size = 3mm}]
    \begin{scope}[every node/.style={state}]
      \node (0) at (0,0) {\footnotesize $r_0$}; \node (1) at (2,0)
      {\footnotesize $r_1$}; \node (2) at (4,0) {\footnotesize $r_1$};
      \node (3) at (8,0) {\footnotesize $r_n$}; \node (4) at (11,0)
      {\footnotesize $t_i$};
    \end{scope}
    \begin{scope}[->,>=stealth, thick]
      \draw (-1,0) to (0); \draw (0) to (1); \draw (1) to (2); \draw
      (2) to (5,0); \draw (6,0) to (3); \draw [bend left = 30] (3) to
      (4); \draw [bend left = 30] (4) to (3);
    \end{scope}
    \node at (1,0.2) {\scriptsize $x\le0$}; \node at (1,-0.2)
    {\scriptsize $release_1!$}; \node at (3,0.2) {\scriptsize
      $x\le0$}; \node at (3,-0.2) {\scriptsize $release_2!$}; \node at
    (7,0.2) {\scriptsize $x\le0$}; \node at (7,-0.2) {\scriptsize
      $release_n!$}; \node at (5.5,0) {\dots}; \node at (9.5,0.7)
    {\scriptsize $done_i?$}; \node at (9.5,0.3) {\scriptsize $\{x\}$};
    \node at (9.5,-0.3) {\scriptsize $x\le0$}; \node at (9.5,-0.7)
    {\scriptsize $release_i!$}; \node at (9.5,1.5) {\vdots}; \node at
    (9.5,-1.5) {\vdots};
  \end{tikzpicture}
  \caption{$\Aa_{wc}$: first release all the tasks in zero time, then
    release a task as soon as the queued instance finishes its
    execution}
  \label{fig:worst-case-task}
\end{figure}
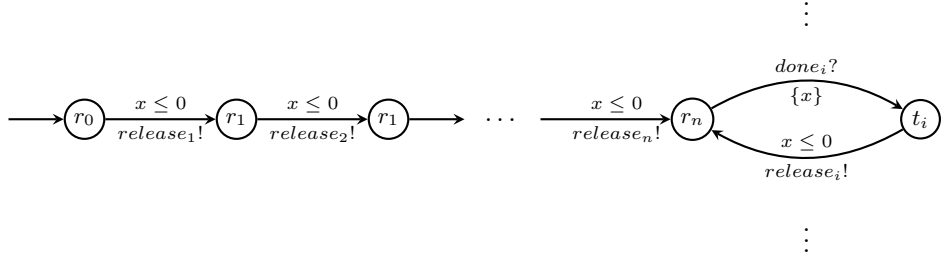

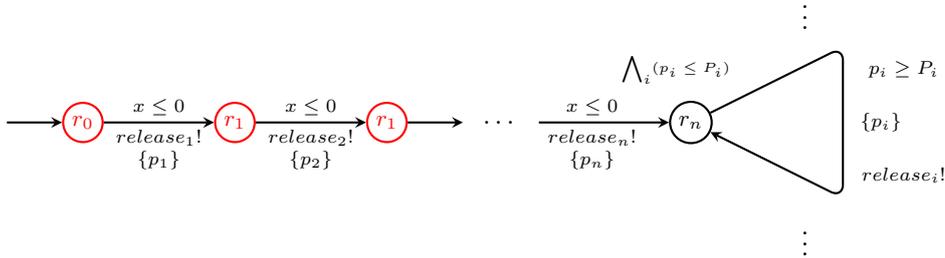
\begin{figure}[h!]
  \centering
  \begin{tikzpicture}[state/.style={circle, thick, draw, inner
      sep=2pt, minimum size = 3mm}]
    \begin{scope}[every node/.style={state}]
      \node [color=red] (0) at (0,0) {\footnotesize $r_0$}; \node
      [color=red] (1) at (2,0) {\footnotesize $r_1$}; \node
      [color=red] (2) at (4,0) {\footnotesize $r_1$}; \node (3) at
      (8,0) {\footnotesize $r_n$};
    \end{scope}
    \begin{scope}[->,>=stealth, thick]
      \draw (-1,0) to (0); \draw (0) to (1); \draw (1) to (2); \draw
      (2) to (5,0); \draw (6,0) to (3); \draw [rounded corners=5] (3)
      to (10,1) to (10,-1) to (3);
    \end{scope}
    \node at (1,0.2) {\scriptsize $x\le0$}; \node at (1,-0.2)
    {\scriptsize $release_1!$}; \node at (1,-0.5) {\scriptsize
      $\{p_1\}$}; \node at (3,0.2) {\scriptsize $x\le0$}; \node at
    (3,-0.2) {\scriptsize $release_2!$}; \node at (3,-0.5)
    {\scriptsize $\{p_2\}$}; \node at (6.7,0.2) {\scriptsize $x\le0$};
    \node at (6.7,-0.2) {\scriptsize $release_n!$}; \node at
    (6.7,-0.5) {\scriptsize $\{p_n\}$}; \node at (7.8,0.7) {\tiny
      $\bigwedge_i(p_i \le P_i)$}; \node at (5.5,0) {\dots}; \node at
    (10.8,0.7) {\scriptsize $p_i \ge P_i$}; \node at (10.5,0)
    {\scriptsize $\{p_i\}$}; \node at (10.8,-0.7) {\scriptsize
      $release_i!$}; \node at (9.5,1.5) {\vdots}; \node at (9.5,-1.5)
    {\vdots};
  \end{tikzpicture}
  \caption{$\Aa_{periodic}$: task releasing automaton for a set of
    periodic tasks $\{t_i \mid i=1,2,\dots,n\}$. The state $r_n$ has
    an invariant $\bigwedge_i(p_i \le P_i)$. The states
    $r_0, \dots, r_{n-1}$ are marked red to denote that they are
    \emph{committed states}. This automaton first releases all tasks
    in zero time, then releases the task $t_i$ after every $P_i$ time
    units, for $i=1,2,\dots,n$, $P_i$ being the period of task $t_i$.}
  \label{fig:periodic-tasks}
\end{figure}

\begin{figure}[h!]
  \centering
  \begin{tikzpicture}[state/.style={rectangle, thick, draw, inner
      sep=2pt, minimum size = 3mm, rounded corners=2}]
    \begin{scope}[every node/.style={state}]
      \node [align=center] (0) at (-0.5,0) {\footnotesize $OFFSET$ \\
        \footnotesize $y \le 60$}; \node [align=center] (1) at (4,0)
      {\footnotesize $LOC\_2$ \\ \footnotesize $x \le 3$}; \node
      [align=center] (2) at (9,0) {\footnotesize $LOC\_3$ \\
        \footnotesize $y \le 60$};
    \end{scope}
    \begin{scope}[->,>=stealth, thick]
      \draw (-2,0) to (0); \draw (0) to (1); \draw (1) to (2); \draw
      [rounded corners=4] (2) to (9,-1) to (4,-1) to (1); \draw
      [rounded corners=4] (1) to (2.5,1.5) to (5.5,1.5) to (1);
    \end{scope}
    \node at (1,0.2) {\scriptsize $y \geq 60$}; \node at (1.8,-0.2)
    {\scriptsize $\{x,y\}~;~n:=0$}; \node at (2.5,0.2) {\scriptsize
      $release_A!$}; \node at (6.5,0.2) {\scriptsize
      $x == 3 \wedge n==N-1$}; \node at (5.5,-0.8) {\scriptsize
      $y \ge 60$}; \node at (7,-0.8) {\scriptsize $release_A!$}; \node
    at (6.5,-1.2) {\scriptsize $\{x,y\}~;~n:=0$}; \node at (4,2)
    {\scriptsize $x==3 \wedge n<N-1$}; \node at (4,1.7) {\scriptsize
      $release_A!$}; \node at (4,1.3) {\scriptsize $\{x\}~;~n:=n+1$};
  \end{tikzpicture}
  \caption{Task automaton releasing the sporadic task for the example
    \emph{SporadicPeriodic} of the tool TIMES. There is one sporadic
    task $(1,3)$ controlled by the above figure. There are three
    periodic tasks $(5,20,20),(8,28,30),(5,30,30)$. We vary the
    variable $N$ giving models SporadicPeriodicN. }
  \label{fig:sporadic-periodic}
\end{figure}
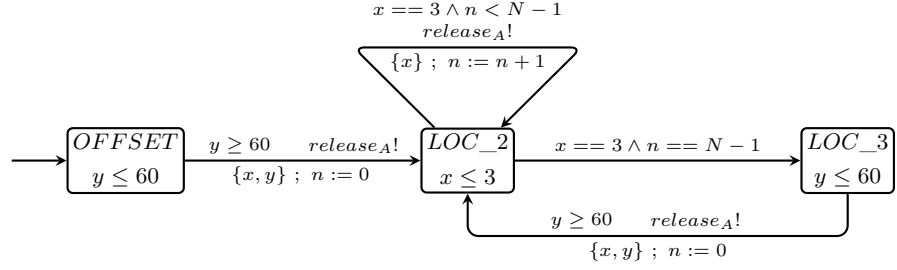

\subsection{Comparison with previous static analysis.}

The automaton $\Aa_{gain}$ present in Figure~\ref{fig:a-gain} is used
to compare the static analysis presented in this paper with the static
analysis presented in \cite{CAV-19}. We check reachability of a
non-existent state, since this forces an entire enumeration of all
nodes modulo simulation and therefore is not sensitive to the order of
exploration of transitions. We use $\Gg^{new}$ to denote the reduced
$\Gg$-map computed by the algorithm of this paper and $\Gg^{old}$ to
denote the $\Gg$-map computed by the algorithm of \cite{CAV-19}.  For
the state $q_2$ of $\Aa_{gain}$,
$\Gg^{old} (q_2) = \{ x-y_1 \le 35, x - y_2 \le 55, x - y_1 \le 30,
x-y_2 \le 65, x \le 50 \}$, whereas
$\Gg^{new}(q_2) = \{x-y_1 \le 35, x-y_1 \le 30, x \le 50\}$. This is
because the propagation of the diagonals $x-y_2 \le 55$ and
$x - y_2 \le 65$ get cut by the guard $x \le 50$ present in the
outgoing transition from $q_2$, due to the optimization presented in
Table~\ref{tab:optimizations-main}. Similarly, for the state $q_1$,
$\Gg^{old}(q_1) = \{ x-y_1 \le 35, x - y_1 \le 30, x \le 20, x \le 40,
x \le 50, x \le 60, x \le 70, x \le 55, x \le 65 \}$, whereas
$\Gg^{new}(q_1) = \{x \le 20\}$. This difference gets amplified when
we consider a product of several copies of $\Aa_{gain}$.  These
reductions in the size of $\Gg$-map indeed show up in the experiments.
As presented in Table~\ref{tab:experiments-new}, the new algorithm
takes lesser time than the old one, while exploring the same size zone
graph. This is because a set $G$ containing more diagonal constraints
results in slower checks for the simulation $\lug$. Since the new
$\Gg$-map does not necessarily provide a coarser simulation than the
old $\Gg$-map, the number of zones enumerated by both the algorithms
are the same, however the new algorithm explores it faster.

\begin{figure}[t]
  \centering
  \begin{tikzpicture}
    \node (0) at (0,0) {\includegraphics[height=100mm,
      width=135mm]{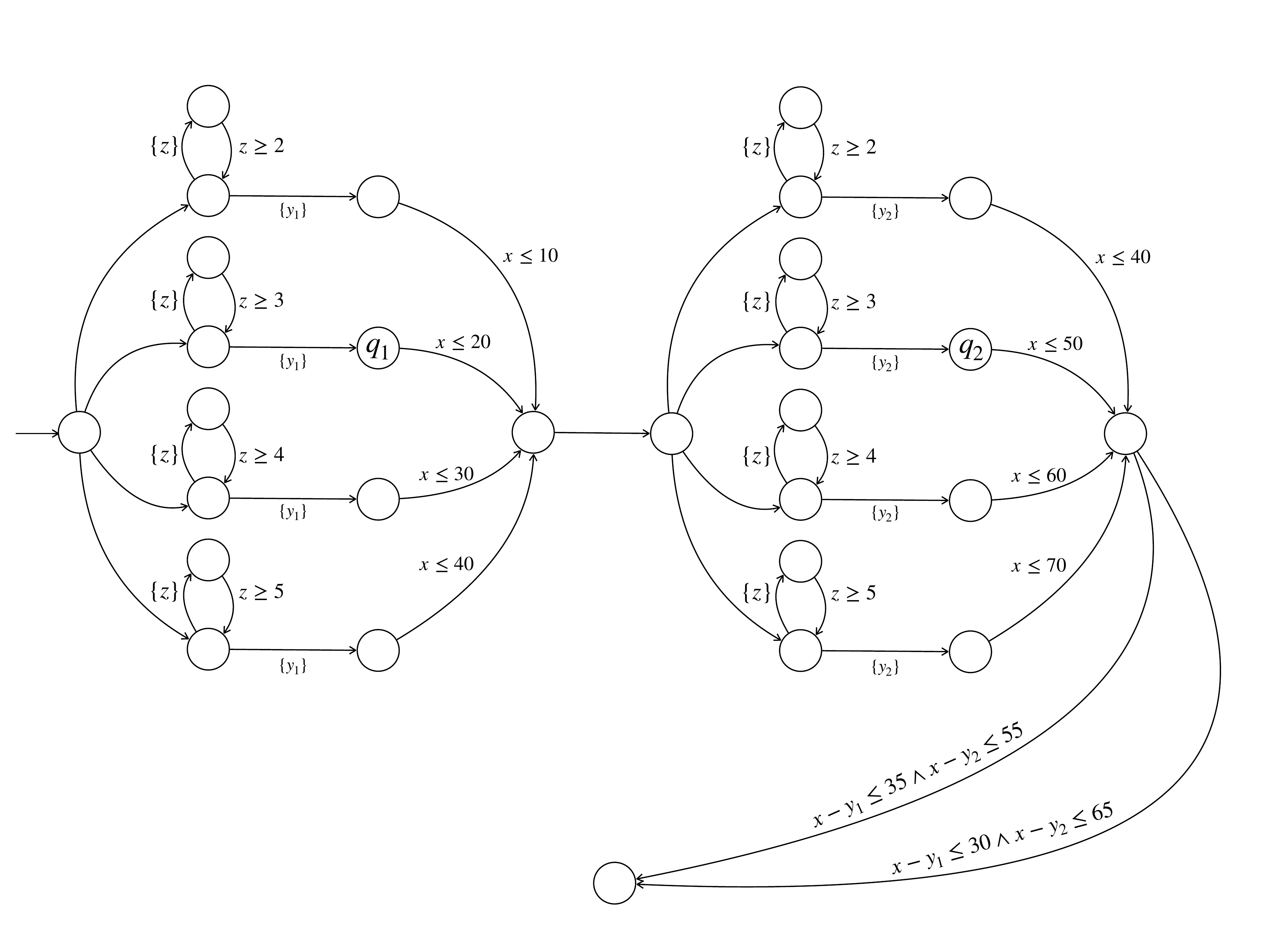}};
  \end{tikzpicture}
  \caption{The automaton $\Aa_{gain}$}
  \label{fig:a-gain}
\end{figure}


\section{Conclusion}
\label{sec:conclusion}

We have presented a static analysis procedure for UTA. Our method
terminates for a wider class of UTA, and hence makes powerful
simulations applicable to this wider class of timed systems. We have
experimented with a prototype implementation. At a technical level, we
get a unifying framework to show decidability that covers the
decidable subclasses of \cite{Bouyer:2004:Updateable}, 
\cite{Fersman:InfComp:2007} and \cite{CAV-19}, which are the only
known decidable classes upto our knowledge and
provides a high-level technique to extend to broader classes: to show
decidability, check if there is a finite reduced $\Gg$-map (c.f. proof
of Theorem \ref{thm:reduced-g-bounds-finite-syntactically-bdd} and the
subsequent remark). Earlier route via regions requires a more involved
low-level reasoning to show the correctness of the region
equivalence. From a practical perspective, we have a prototype with a
richer modeling language and a more efficient way to handle updates
than the existing real-time model checkers. As future work, we plan to
engineer the prototype to make it applicable for bigger models and
release the implementation and benchmarks in the public domain.

\newpage
\bibliographystyle{plainurl} \bibliography{UTA-better-static-analysis}

\end{document}